\documentclass{amsart}

\RequirePackage[OT1]{fontenc}
\RequirePackage{amsthm,amsmath}
\RequirePackage[numbers]{natbib}
\RequirePackage[colorlinks,citecolor=blue,urlcolor=blue]{hyperref}
\usepackage{graphicx}
\usepackage{latexsym,color}
\usepackage{graphicx}
\usepackage{amssymb}
 \usepackage{amsfonts}
 \usepackage{amsmath}
\usepackage{amsthm}
\newtheorem{thm}{Theorem}[section]
 
 \newtheorem{lem}[thm]{Lemma}

\newtheorem{asm}[thm]{Assumption}

\newtheorem{remark}[thm]{Remark}
\newtheorem{example}[thm]{Example}

\numberwithin{equation}{section}

\newcommand{\cB}{\mathcal{B}}

 \newcommand{\cP}{\mathcal{P}}

\DeclareMathOperator{\law}{Law}

\newcommand{\sigmalow}{\underline{\sigma}}
\newcommand{\sigmahigh}{\overline{\sigma}}

\def \F{\mathbb{F}}

\newcommand{\set}{\triangleq}

\numberwithin{equation}{section}
\theoremstyle{plain}

\begin{document}

\title{Super-replication with nonlinear transaction costs and
  volatility uncertainty}
 \thanks{The authors were supported by the Einstein Foundation Grant no.A 2012 137. The second author also
acknowledges support of the Marie Curie Actions fellowships Grant no.618235.
}
\author{Peter Bank, Yan Dolinsky and Selim G\"okay \\
 Hebrew University and TU Berlin
}%
\address{
 Department of Mathematics, TU Berlin\\
 {e.mail: bank@math.tu-berlin.de}}
\address{
 Department of Statistics, Hebrew University of Jerusalem\\
 {e.mail: yan.yolinsky@mail.huji.ac.il}}

\address{
 Department of Mathematics, TU Berlin\\
 {e.mail: gokay@math.tu-berlin.de}}

\date{\today}

\begin{abstract}
We study super-replication of contingent claims in an illiquid
  market with model uncertainty. Illiquidity is captured by nonlinear
  transaction costs in discrete time and model uncertainty arises as
  our only assumption on stock price returns is that they are in a
  range specified by fixed volatility bounds. We provide a dual
  characterization of super-replication prices as a supremum of
  penalized expectations for the contingent claim's payoff. We also
  describe the scaling limit of this dual representation when the
  number of trading periods increases to infinity. Hence, this paper
  complements the results in~\cite{DS1} and~\cite{K} for the case of
  model uncertainty.
\end{abstract}

\maketitle

\markboth{P.Bank, Y.Dolinsky and S.G\"okay}{Nonlinear transaction costs and volatility uncertainty}
\renewcommand{\theequation}{\arabic{section}.\arabic{equation}}
\pagenumbering{arabic}

\section{Introduction}\label{sec:1}
We study an illiquid discrete--time market with model uncertainty.  As
in \cite{DS1} we consider the case where the size of the trade has an
immediate but temporary effect on the price of the asset.  This model
captures the classical case of proportional transaction costs as well
as other illiquidity models such as the discrete-time version of the
model introduced by Cetin, Jarrow and Protter in \cite{CJP} for
continuous time. By contrast to \cite{DS1}, our sole assumption on the
price dynamics of the traded security is that the absolute value of the
log-returns is bounded from below and above. This is a natural
discrete-time version of the widely studied uncertain volatility
models; see, e.g., \cite{DM1}, \cite{NS} and \cite{STZ}. The paper
\cite{D} studies super-replication of game options in such a
discrete-time model, but does not allow for any market frictions.

The benchmark problem in models with uncertain volatilities is the
description of super-replication prices. In our version of such a
result in a model with transaction costs, we provide a general duality
for European options with an upper-semicontinuous
payoff. Specifically, Theorem~\ref{thm2.1} provides a combination of
the dual characterization for super-replication prices in
frictionless uncertain volatility models (see \cite{DM1}) with
analogous duality formulae in binomial markets with frictions (see
\cite{DS1}).  Let us emphasize that, under volatility uncertainty, we
do not define super-replication prices merely in an almost
sure--sense because we actually insist on super-replication for
\emph{any} possible evolution of stock prices which respects the
specified volatility bounds.  Since, in contrast to \cite{DS1}, in our
setup the set of all possible stock price evolutions is uncountable we
cannot find a dominating probability measure which will give positive
weight to every possible evolution of stock prices. Hence, the duality
result of Theorem~\ref{thm2.1} goes beyond the now classical results
for almost sure super-replication, and it complements the findings of
\cite{BN}, who consider a probabilistically more general setting albeit
without frictions. For a discussion of the fundamental theorem of
asset pricing with proportional transaction costs under model
uncertainty, see \cite{BN1}.  The key observation for the proof of
Theorem~\ref{thm2.1} is that for continuous payoffs one can find
approximative discrete models where classical duality results give us
super-replication strategies that can be lifted to the original
setting with uncertain volatilities in a way that allows us to control
the difference in profits and losses.

In Theorem~\ref{thm2.2} we consider the special case of a convex
payoff profile. In frictionless models with volatility uncertainty it
is well-known from, e.g., \cite{L}, \cite{H} or Remark~7.20 in~\cite{FSbook}
that the super-replication price coincides with the one computed in
the classical model where the volatility always takes the maximal
value. We show that this result also holds in our framework with
nonlinear transaction costs if these are deterministic.

Finally, we study the scaling limit of our super-replication prices
when the number of trading periods becomes large. Theorem~\ref{thm2.3}
describes this scaling limit as the value of a stochastic volatility
control problem on the Wiener space.  We use the duality result
Theorem~\ref{thm2.1} and study the limit of the dual terms by applying
the theory of weak convergence of stochastic processes.  This extends
the approach taken in \cite{DS1}, since we need to extend Kusuoka's
construction of suitable martingales (see \cite{K}, Section 5) to our
setting with volatility uncertainty. As a result, we get an
understanding of how Kusuoka's finding of volatility uncertainty in
the description of the scaling
limit is extended to our setting which already starts with
volatility uncertainty. In the special case of a
frictionless setup, we recover Peng's \cite{P} result that the
limit is equal to the payoff's $G$--expectation with upper and lower
bounds as in the discrete setup.  In setups with market frictions,
continuous-time models are known to produce trivial super-replication
prices when one considers proportional transaction costs, cf., e.g.,
~\cite{LS} and~\cite{SSC}, or no liquidity effect at all when one has
nonlinear, differentiable transaction costs or market impact,
see~\cite{CJP} or~\cite{BB}. Our scaling limit, by contrast, gives a
value in between these two extremes and can be viewed as a convex
measure of risk for the payoff as in~\cite{FS} or~\cite{FG}.

Our approach to the proof of the main results is purely probabilistic
and based on the theory of weak convergence of stochastic
processes. This approach allows us to study a quite general class of
path dependent European options, and a general class of nonlinear
transaction costs.

\section{Preliminaries and Main Results}
\label{sec:2}
\subsection{The discrete-time model}

Let us start by introducing a discrete time financial model with
volatility uncertainty. We fix a time horizon $N \in \mathbb N$
and consider a financial market with a riskless savings account and a
risky stock. The savings account will be used as a numeraire and thus
we normalize its value at time $n = 0,\dots,N$ to $B_n = 1$. The stock
price evolution starting from $s>0$ will be denoted by $S_n>0$,
$n=0,1,\dots,N$. Hence, by introducing the log-return $X_n \set
\log(S_n/S_{n-1})$ for period $n=1,\dots,N$ we can write
\begin{equation}
  \label{eq:1}
  S_n = s_0 \exp\left(\sum_{m=1}^n X_m\right), \quad n=0,\dots,N.
\end{equation}
Our sole assumption on these dynamics will be that there are
volatility bounds on the stock's price fluctuations in the sense that
the absolute values of these log-returns are bounded from below and
above:
\begin{equation}
  \label{eq:2}
  \sigmalow \leq |X_n| \leq \sigmahigh, \; n=1,\dots,N,
\end{equation}
for some constants $0 \leq \sigmalow \leq \sigmahigh <\infty$.  In
other words the log-returns will take values in the path-space
$$
\Omega \set \Omega_{\sigmalow,\sigmahigh} \set
\left\{\omega=(x_1,\dots,x_N) \in \mathbb{R}^N \;:\; \sigmalow \leq |x_n| \leq
  \sigmahigh, \; n=1, \dots, N \right\}
$$
and identifying these returns with the canonical process
$$
X_k(\omega) \set x_k \text{ for } \omega=(x_1,\dots,x_N)\in\Omega
$$
we find that \eqref{eq:1} allows us to view the stock's price
evolution as a process $S=(S_n)_{n=0,\dots,N}$ defined on
$\Omega$. Clearly, the canonical filtration
\begin{equation*}
  \mathcal{F}_n \set \sigma(X_1,\dots,X_n), \quad n=0,\dots,N,
\end{equation*}
coincides with the one generated by $S=(S_n)_{n=0,\dots,N}$.  Similar
models with volatility uncertainty have been considered in \cite{D}.

The aim of the present paper is to study the combined effects of
volatility uncertainty and nonlinear transaction costs. Following
\cite{CJP, DS1, GS}, we assume
these costs to be given by a penalty function
$$
g:\{0,1,...,N\}\times \Omega \times\mathbb R\rightarrow\mathbb R_{+}
$$
where $g(n,\omega,\beta)$ denotes the costs (in terms of our numeraire
$B$) of trading $\beta \in \mathbb{R}$ worth of stock at time $n$ when
the evolution of the stock price is determined by the returns from
$\omega \in \Omega$.

\begin{asm}
\label{a.g}
The cost function
$$
g:\{0,1,...,N\} \times \Omega\times\mathbb R\rightarrow\mathbb R_{+}
$$
is $(\mathcal{F}_n)_{n=0,\dots,N}$-adapted. Moreover, for any
$n=0,\dots,N$, the costs
$g(n,\omega,\beta)$ are a nonnegative
convex function in $\beta \in \mathbb{R}$ with $g(n,\omega,0)=0$ for any fixed
$\omega \in \Omega$ and a continuous function in $\omega \in \Omega$
for any fixed $\beta \in \mathbb{R}$.
\end{asm}

For simplicity of notation we will often suppress the dependence of
costs on $\omega$ and simply write $g_n(\beta)$ for
$g(n,\omega,\beta)$. We will proceed similarly with other functions
depending on $\omega \in \Omega$.

In our setup a trading strategy is a pair $\pi=(y,\gamma)$ where $y$
denotes the initial wealth and $\gamma:\{0,1,...,N-1\}\times
\Omega\rightarrow\mathbb{R}$ is an $(\mathcal{F}_n)$-adapted process
specifying the number $\gamma_n=\gamma(n,\omega)$ of shares held at
the beginning of any period $n=0,\dots,N-1$ with the stock price
evolution given by $\omega \in \Omega$. The set of all portfolios
starting with initial capital $y$ will be denoted by $\mathcal{A}(y)$.

The evolution of the mark-to-market value $Y^\pi =
(Y^\pi_n(\omega))_{n=0,\dots,N}$ resulting from a trading strategy
$\pi=(y,\gamma) \in \mathcal{A}(y)$ is given by $Y^{\pi}_0=y$ and the
difference equation
\begin{equation*}
\label{2.1}
Y^{\pi}_{n+1}- Y^\pi_n=
\gamma_n(S_{n+1}-S_n)-g_n((\gamma_n-\gamma_{n-1})S_{n}), \quad n=0,\dots,N-1,
\end{equation*}
where we let $\gamma_{-1}\set 0$. Hence, we start with zero stocks in
our portfolio and trading to the new position $\gamma_n$ to be held
after time $n$ incurs the transaction costs
$g_n((\gamma_n-\gamma_{n-1})S_n)$, which is the only friction in our
model. Hence, the value $Y^{\pi}_{n+1}$ represents the portfolio's
mark-to-market value \emph{before} the transaction at time $n+1$ is
made. Note that, focussing on the mark-to-market value rather than the
liquidation value, we disregard in particular the costs of unwinding
any non-zero position for simplicity.

\subsection{Robust super-replication with frictions}

The benchmark problem for models with uncertain volatility is the
super-replication of a contingent claim. We investigate this problem
in the presence of market frictions as specified by a function $g$
satisfying Assumption~\ref{a.g}. So consider a European option
$\F:\mathbb{R}_+^{N+1} \to \mathbb{R}_+$ which pays off $\F(S)$ when the
stock price evolution is $S=(S_n)_{n=0,\dots,N}$. The
super-replication price $V(\F) = V^{\sigmalow,\sigmahigh}_{g}(\F)$ is
then defined as
\begin{equation*}\label{2.5}
V^{\sigmalow,\sigmahigh}_{g}(\F) \set \inf\left\{y \in \mathbb{R} \;:\; \exists\pi\in\mathcal{A}(y)
  \text{ with }
Y^\pi_N(\omega)\geq \F(S(\omega)) \ \forall \omega\in \Omega_{\sigmalow,\sigmahigh}\right\}.
\end{equation*}
We emphasize that we require the construction of a \emph{robust}
super-replication strategy $\pi$ which leads to a terminal value
$Y^\pi_N$ that dominates the payoff $X$ in \emph{any} conceivable
scenario $\omega \in \Omega$.

Our first result provides a dual description of super-replication prices:

\begin{thm}\label{thm2.1}
Let $G_n:\Omega \times \mathbb{R}\rightarrow\mathbb{R}_{+}$, $n=0,1,...,N$,
denote the Legendre-Fenchel transform (or convex conjugate) of $g_n$, i.e.,
\begin{equation*}\label{2.6--}
\begin{split}
G_n(\alpha)\set\sup_{\beta\in\mathbb{R}}\left\{\alpha \beta -g_n(\beta)\right\},
\quad \alpha \in\mathbb{R}.
\end{split}
\end{equation*}

Then, under Assumption~\ref{a.g}, the super-replication price of any
contingent claim $\F$ with upper-semicontinuous payoff function
$\F:\mathbb{R}_+^{N+1} \to \mathbb{R}_+$ is given by
\begin{equation*}
  V(\F)=\sup_{\mathbb P\in \mathcal{P}_{\sigmalow,\sigmahigh}}\mathbb E_{\mathbb P}\left[\F(S)-\sum_{n=0}^{N-1}
      G_n\left(\frac{\mathbb{E}_{\mathbb  P}[S_N\;|\;\mathcal{F}_n]-S_n}{S_n}\right)  \right]
\end{equation*}
where $\mathcal{P}_{\sigmalow,\sigmahigh}$ denotes the set of all
Borel probability measures on $\Omega=\Omega_{\sigmalow,\sigmahigh}$
and where $\mathbb{E}_{\mathbb P}$ denotes the expectation with
respect to such a probability measure $\mathbb P$.
\end{thm}
The proof of this theorem will be carried out in
Section~\ref{sec:discretedualityproof} below.  Observe that this
result is a hybrid of the dual characterization for super-replication
prices in frictionless uncertain volatility models and of analogous
duality formulae in binomial markets with frictions; see \cite{DM} and
\cite{DS1}, respectively.

\subsection{Convex payoff functions}

Our next result deals with the special case where the payoff $\F$ is a
nonnegative \emph{convex} function of the stock price evolution
$S=(S_n)_{n=0,\dots,N}$. It is well known that in a frictionless
binomial model, the price of a European option with a convex payoff
is an increasing function of the volatility. This implies that
super-replication prices in uncertain volatility models coincide with
the replication costs in the model with maximal compatible
volatility; see \cite{L}. The next theorem gives a
generalization of this claim for the setup of volatility uncertainty
under friction.

\begin{thm}\label{thm2.2}
  Suppose that the cost function $g=g(n,\omega,\beta)$ is
  deterministic in the sense that it does not depend on $\omega \in
  \Omega$.

  Then the super-replication price of any convex payoff
  $\F:\mathbb{R}_+^{N+1} \to \mathbb{R}_+$ is given by
  $$
  V^{\sigmalow,\sigmahigh}_g(\F)=\overline{V}_g(\F)
  $$
  where $\overline{V}_g(\F) \set V^{\sigmahigh,\sigmahigh}_g(\F)$
  denotes the super-replication price of $\F=\F(S)$ in the Binomial model with frictions for $S$ with volatility $\sigmahigh$ and cost
  function $g$.
\end{thm}
\begin{proof}
  The relation `$\geq$' holding true trivially, it suffices to
  construct, for any $\epsilon>0$, a strategy $\gamma$ which
  super-replicates $\F(S)$ in every scenario from $\Omega$ starting
  with initial capital $y=\epsilon+\overline{V}_g(\F)$.

  The binomial model with volatility $\sigmahigh$ can be formalized on
  $\overline{\Omega} \set \{-1,1\}^N$ with canonical process
  $\overline{X}_k(\overline{\omega})\set\overline{x}_k$ for
  $\overline{\omega}=(\overline{x}_1,\dots,\overline{x}_N) \in
  \overline{\Omega}$ by letting the stock price evolution be given
  inductively by $\overline{S}_0 \set s_0$ and $\overline{S}_n \set
  \overline{S}_{n-1} \exp(\sigmahigh \overline{X}_n)$,
  $n=1,\dots,N$. With $(\overline{\mathcal{F}}_n)_{n=0,\dots,N}$
  denoting the corresponding canonical filtration, we get from the
  definition of $\overline{V}_g(\F)$ that there is an
  $(\overline{\mathcal{F}}_n)_{n=0,\dots,N}$-adapted process
  $\overline{\gamma}$ such that with $\overline{\gamma}_{-1}\set 0$ we
  have
  \begin{equation}
    \label{eq:3}
    \epsilon+\overline{V}_g(\F)+\sum_{n=0}^{N-1} \overline{\gamma}_n
    (\overline{S}_{n+1}-\overline{S}_n) - \sum_{n=0}^{N-1}
    g_n((\overline{\gamma}_n-\overline{\gamma}_{n-1})\overline{S}_n)
    \geq \F(\overline{S})
  \end{equation}
  everywhere on $\overline{\Omega}$.

  In view of~\eqref{eq:2}, for any $\omega \in \Omega$ and
  $n=1,\dots,N$ there are unique weights $\lambda_n^{(+1)}(\omega),
  \lambda_n^{(-1)}(\omega) \geq 0$ with
  $\lambda_n^{(+1)}(\omega)+\lambda_n^{(-1)}(\omega)=1$ such that
  \begin{equation}\label{eq:4}
    e^{X_n(\omega)} = \lambda_n^{(+1)}(\omega) e^{\sigmahigh} +
    \lambda_n^{(-1)}(\omega)e^{-\sigmahigh}.
  \end{equation}
  It is readily checked that also the weights
  \begin{align*}
    \lambda^{\overline{\omega}^n}_n(\omega) \set\prod_{m=1}^n
    \lambda_m^{(\overline{\omega}^n_m)}(\omega),
  \quad \overline{\omega}^n=(\overline{\omega}^n_1,\dots,\overline{\omega}^n_n)
  \in \{-1,+1\}^n,
    \end{align*}
   sum up to $1$. Indeed,
  \begin{equation}
    \label{eq:5}
    \sum_{\overline{\omega}^n \in \{-1,+1\}^n}
    \lambda^{\overline{\omega}^n}_n =\prod_{m=1}^n \left(\lambda_m^{(+1)}(\omega)+\lambda_m^{(-1)}(\omega)\right)=1, \quad n=1,\dots,N.
  \end{equation}
  Moreover, for $n=1,...,N-1$ we have
  \begin{equation}
    \label{eq:6}
    \lambda_n^{\overline{\omega}^n} =\lambda_n^{\overline{\omega}^n}\prod_{m=n+1}^N \left(\lambda_m^{(+1)}(\omega)+\lambda_m^{(-1)}(\omega)\right)=
    \sum_{\overline{\omega}^{N-n} \in \{-1+,1\}^{N-n}}
    \lambda_N^{(\overline{\omega}^n,\overline{\omega}^{N-n})},
  \end{equation}
  which in conjunction with~\eqref{eq:4} and the adaptedness of
  $\overline{S}$ entails the representation
  \begin{align}
    \label{eq:7}
    S_n(\omega) &=s_0\prod_{m=1}^n \left(\lambda_m^{(+1)}(\omega) e^{\sigmahigh} +
    \lambda_n^{(-1)}(\omega)e^{-\sigmahigh}\right)\\\nonumber
    &=  \sum_{\overline{\omega}^n \in \{-1,+1\}^n}
    \overline{S}_n(\overline{\omega}^n,1,\dots,1)
    \lambda_n^{\overline{\omega}^n}(\omega)
    = \sum_{\overline{\omega} \in \overline{\Omega}}
    \overline{S}_n(\overline{\omega})
    \lambda_N^{\overline{\omega}}(\omega)\nonumber
  \end{align}
 for any $n=1,\dots,N$, $\omega \in \Omega$.

  Now evaluate~\eqref{eq:3} at $\overline{\omega} \in
  \overline{\Omega}$, multiply by $\lambda^{\overline{\omega}}_N$ and
  then take the sum over all $\overline{\omega} \in
  \overline{\Omega}$.

  The right side of~\eqref{eq:3} then aggregates to
  \begin{align}\label{eq:8}
    R & \set \sum_{\overline{\omega} \in \overline{\Omega}}
    \F(\overline{S}(\overline{\omega}))\lambda_N^{\overline{\omega}}
 \geq \F\left(\sum_{\overline{\omega} \in \overline{\Omega}}
    \overline{S}(\overline{\omega})\lambda_N^{\overline{\omega}}\right)=\F(S)
  \end{align}
  where the estimate follows from~\eqref{eq:5} in conjunction with the
  convexity of $\F:\mathbb{R}_+^{N+1} \to \mathbb{R}$ and where for the
  last identity we exploited~\eqref{eq:7}.

  On the left side of~\eqref{eq:3} the contributions from the constant
  $\epsilon+\overline{V}_g(\F)$ just reproduce this very constant because
  of~\eqref{eq:5}. Since $\overline{\gamma}$ and $\overline{S}$ are
  $(\overline{\mathcal{F}}_n)_{n=0,\dots,N}$-adapted, the $n$th
  summand in the first sum of~\eqref{eq:3} contributes
  \begin{align*}
    I_n & \set \sum_{\overline{\omega} \in \overline{\Omega}}
    \overline{\gamma}_n(\overline{\omega})
    (\overline{S}_{n+1}(\overline{\omega})-\overline{S}_n(\overline{\omega}))\lambda_N^{\overline{\omega}} \\
    & = \sum_{\overline{\omega}^n \in \{-1,1\}^n} \sum_{\overline{x}
      \in \{-1,1\}} (\overline{\gamma}_n \overline{S}_n)
    (\overline{\omega}^n,1,\dots,1)(e^{\overline{\sigma}\overline{x}}-1)
    \lambda^{\overline{\omega}^n}_n \lambda^{(\overline{x})}_{n+1}.
  \end{align*}
 By definition of $\lambda_{n+1}^{(\pm 1)}$ we have
 \begin{equation*}
   \sum_{\overline{x}
      \in \{-1,1\}}
    (e^{\overline{\sigma}\overline{x}}-1)\lambda^{(\overline{x})}_{n+1}=e^{X_{n+1}}-1 = (S_{n+1}-S_n)/S_n
 \end{equation*}
 which entails
 \begin{align*}
   \label{eq:9}
   I_n &= \sum_{\overline{\omega}^n \in \{-1,1\}^n} (\overline{\gamma}_n \overline{S}_n)
    (\overline{\omega}^n,1,\dots,1) \lambda^{\overline{\omega}^n}_n (S_{n+1}-S_n)/S_n
    \\
   & = \gamma_n(S_{n+1}-S_n)
 \end{align*}
 where the $(\mathcal{F}_n)_{n=0,\dots,N}$-adapted process $\gamma$ is
 given by $\gamma_0 \set \overline{\gamma}_0$ and
 \begin{equation*}
   \gamma_n(\omega) = \sum_{\overline{\omega}^n \in \{-1,1\}^n}
   (\overline{\gamma}_n \overline{S}_n)
   (\overline{\omega}^n,1,\dots,1)\lambda_n^{\overline{\omega}^n}(\omega)/S_n(\omega),
   \quad \omega \in \Omega,
 \end{equation*}
 for $n=1,\dots,N-1$. In a similar fashion the $n$th summand from the
 second sum in~\eqref{eq:3} gives
 \begin{align*}
   I\!I_n & \set \sum_{\overline{\omega} \in \overline{\Omega}}
   g_n((\overline{\gamma}_n(\overline{\omega})-
   \overline{\gamma}_{n-1}(\overline{\omega}))\overline{S}_n(\overline{\omega}))\lambda_N^{\overline{\omega}}
   \\
   & \geq g_n\left(\sum_{\overline{\omega} \in \overline{\Omega}}
     (\overline{\gamma}_n(\overline{\omega})-
     \overline{\gamma}_{n-1}(\overline{\omega}))
     \overline{S}_n(\overline{\omega})\lambda_N^{\overline{\omega}}\right)
   \\
   & = g_n((\gamma_n-\gamma_{n-1})S_n)
 \end{align*}
 where the estimate is due to~\eqref{eq:5} and the convexity of
 $g_n:\mathbb{R} \to \mathbb{R}$ and where the last identity is due to
 the adaptedness of $\overline{\gamma}$, $\overline{S}$ and
 to~\eqref{eq:6}. As a consequence, the left side of~\eqref{eq:3}
 aggregates in the above manner to
 \begin{align*}
  L &\set \epsilon+\overline{V}_g(\F) + \sum_{n=0}^{N-1} I_n - \sum_{n=0}^{N-1} I\!I_n
   \\
& \leq  \epsilon+\overline{V}_g(\F) + \sum_{n=0}^{N-1} \gamma_n(S_{n+1}-S_n)- \sum_{n=0}^{N-1}g_n((\gamma_n-\gamma_{n-1})S_n).
 \end{align*}
 In light of our estimate~\eqref{eq:8} for the analogously aggregated
 right side of~\eqref{eq:3} this shows that $\gamma$ super-replicates
 $\F(S)$ with initial capital $\epsilon+\overline{V}_g(\F)$. This
 accomplishes our proof.
\end{proof}

\subsection{Scaling limit}\label{sec:3}

Our last result gives a dual description for the scaling limit of our
super-replication prices when the number of periods $N$ becomes large,
stock returns are scaled by $1/\sqrt{N}$ and earned over periods of
length $1/N$. The limiting trading costs will be specified in terms of
a function
\begin{align*}
h:[0,1] \times C[0,1]\times\mathbb R&\rightarrow\mathbb R_{+}\\
(t,w,\beta) \mapsto h_t(w,\beta)
\end{align*}
such that
\begin{itemize}
\item for any $t \in [0,1]$, $w \in C[0,1]$, $\beta \mapsto
h_t(w,\beta)$ is nonnegative and convex with $h_t(w,0)=0$;

\item for any $\beta \in \mathbb{R}$ the process
$h(\beta)=(h_t(w,\beta))_{t \in [0,1]}$ is progressively measurable in
the sense that $h_t(w,\beta)=h_t(\tilde w,\beta)$
if $w_{[0,t]}=\tilde w_{[0,t]}$.
\end{itemize}

For technical reasons we were able to establish our scaling limit
  only for a suitably truncated penalty function. For a given
  truncation level $c>0$, we consider the linearly extrapolated costs $h^c$ given by
\begin{equation*}
  h^c_t(w,\beta) \set
 \begin{cases}
  h_t(w,\beta) \text{ for } \beta \in I_c(h)   \\
 \text{linear extrapolation with slope $c$ beyond } I_c(h)
 \end{cases}
\end{equation*}
where
$I_c(h) \set \left\{\beta \in \mathbb{R} \;:\; \left|\frac{\partial
      h}{\partial \beta}\right| \leq c\right\}$
denotes the interval around zero where the slope of $h$ has not yet
exceeded $c$ in absolute value. For a truncated penalty $h^c$ the dual
formula of super-replication prices will only involve measures which
satisfy a certain tightness condition; this allows us to establish the
upper bound of our scaling limit in Section \ref{sec:leq}. For the
original $h$ (by taking $c$ to infinity) we can only prove the lower
bound of the super-replication prices, the proof for a tight upper
bound remaining an open problem in this case.

The cost for the $N$ period model with returns in
\begin{equation*}
  \Omega^N \set \{\omega^N=(x_1,\dots,x_N) \;:\; \sigmalow/\sqrt{N}
  \leq |x_n| \leq \sigmahigh/\sqrt{N}, n=1,\dots,N\}
\end{equation*}
are given by
\begin{equation}\label{eq:10}
  g^{N,c}_{n}(\omega^N,\beta) \set h_{n/N}^{c/\sqrt{N}}(\overline{S}^N(\omega^N),\beta)
\end{equation}
where $\overline{S}^N(\omega^N) \in C[0,1]$ denotes the continuous linear interpolation
in $C[0,1]$ of the points
\begin{equation*}
  \overline{S}^N _{n/N} (\omega^N) \set s_0 \exp\left(\sum_{m=1}^n x_m\right), \; n=0,\dots,N,
\end{equation*}
for  $\omega^N=(x_1,\dots,x_N) \in \Omega^N$.

The technical assumption for our asymptotics to work out is the
following:
\begin{asm}
\label{quadratic}
The Legendre-Fenchel transform $H:[0,1] \times C[0,1]\times\mathbb
R\rightarrow\mathbb R$ with
$$
H_t(w,\alpha)=\sup_{\beta\in\mathbb R}
\left\{\alpha\beta -h_t(w,\beta)\right\}, \quad t \in [0,1], \; w \in
C[0,1], \; \alpha \in \mathbb{R}
$$
has polynomial growth in $(w,\alpha)$ uniformly in $t$ in the
sense that there are constants $C,p_1,p_2\geq 0$ such that
\begin{equation*}
  H_t(w,\alpha)\leq C(1+||w||^{p_1}_{\infty})(1+|\alpha|^{p_2}), \quad
  (t,w,\alpha)\in [0,1] \times C[0,1]\times\mathbb{R}.
\end{equation*}
In addition $H$ is continuous in $(t,w)$ and essentially quadratic in
$\alpha$ asymptotically, i.e., there is a function $\hat{H}:[0,1]
\times \mathcal C[0,1]\rightarrow\mathbb R_{+}$ such that for any
sequence ${\{(t_N,w_N,\alpha_N)\}}_{N=1,2,\dots}$ converging to
$(t,w,\alpha)$ in $[0,1] \times C[0,1]\times\mathbb{R}$ we
have
$$
\lim_{N\rightarrow\infty} |N H_{t_N}(w_N,\alpha_N/\sqrt N)-\hat{H}_t(w)\alpha^2|=0.
$$
\end{asm}

Let us give two examples.
\begin{example}
\begin{enumerate}
\item[(i)] Proportional transaction costs:
Fix $c>0$ and consider the cost functions given by
$$g^{N,c}_{n}(\omega^N,\beta) \set \frac{c}{\sqrt N}|\beta| $$
for our binomial market model with $N$ trading periods. This example
was studied in \cite{K} for binomial models and
corresponds in our setup to the case
$\hat H=H\equiv 0$.

\item[(ii)] Quadratic costs. For a given constant $\Lambda >0$, let
$$
h_t(w,\beta) = \Lambda \beta^2.
$$
Fix $c>0$ and observe that our truncation $h^c$ of $h$ is then
$$
h^c(w,\beta)=
 \begin{cases}
\Lambda \beta^2  & {\mbox{if}}\ |\beta| \le
\frac{c}{2\Lambda},\\
c \beta -
\frac{c^2 }{4 \Lambda}, &{\mbox{else.}}
\end{cases}
$$
Thus the penalty in the $N$--step binomial model is given by
$$
g^{N,c}_n(\omega^N,\beta)=
\left\{
\begin{array}{ll}
\Lambda \beta^2 ,\qquad &{\mbox{if}}\ |\beta| \le
\frac{c}{2\Lambda\sqrt N},\\
\frac{c}{\sqrt N} |\beta| -
\frac{c^2 }{4 \Lambda N}, &{\mbox{else.}}
\end{array}
\right.
$$
Hence the marginal costs from trading slightly higher volumes are
linear for a small total trading volume and constant for a large one.
This example corresponds to the case where
$H_t(w,\beta)=\frac{\beta^2}{4\Lambda}$ and
$\hat H_t(w)\equiv\frac{1}{4\Lambda}$.
\end{enumerate}
\end{example}

\begin{remark}
  A sufficient condition for a limiting cost process $h$ to satisfy
  our Assumption~\ref{quadratic} is the joint validity of
 \begin{enumerate}
 \item[(i)] there exists $\epsilon>0$ such that for any
   $(t,w)\in [0,T]\times C^{+}[0,T]$, the second derivative
   $\frac{\partial^2 h}{\partial\beta^2}(t,w,\beta)$ exists for any
   $-\epsilon w(t)<\beta<\epsilon w(t)$, and is continuous at
   $(t,w,0)$, and

 \item[(ii)] we have
   \begin{equation*}
     \begin{split}
       \frac{\partial h}{\partial\beta}(t,w,0)\equiv 0, \ \ \mbox{and}
       \ \ \inf_{(t,w)\in [0,T]\times
         C^{+}[0,T]}\inf_{|\beta|<\epsilon w(t)} \frac{\partial^2
         h}{\partial\beta^2}(t,w,\beta)>0.
     \end{split}
   \end{equation*}
 \end{enumerate}
 Indeed, sufficiency of these conditions can be verified by use of a
 Taylor expansion.
\end{remark}

Under Assumption~\ref{quadratic} the scaling limit for the discrete-time
super-replication prices can be described as follows:

\begin{thm}\label{thm2.3}
  Suppose that Assumption~\ref{quadratic} holds and that
  $\sigmalow>0$. Furthermore assume  that
  $\F:C[0,1] \to \mathbb{R}_+$ is continuous with polynomial growth:
  $ 0 \leq \F(S) \leq C(1+\|S\|_{\infty})^p$, $S \in C[0,1]$, for some
  constants $C,p \geq 0$.

  Then we have
  \begin{equation*}
    \lim_{N\rightarrow\infty}V^{\sigmalow/\sqrt{N},\sigmahigh/\sqrt{N}}_{g^{N,c}}(\F)
= \sup_{\sigma \in \Sigma(c)} \mathbb{E}^W
         \left[\mathbb F(S^{\sigma})-\int_{0}^1\hat H_t(S^{\sigma})a^2(\sigma_t) dt\right]
  \end{equation*}
where
\begin{itemize}
\item $\mathbb{E}^W$ denotes the expectation with respect to
  $\mathbb{P}^W$, the Wiener measure on $(C[0,1],\cB(C[0,1]))$, for
  which the canonical process $W$ is a Brownian motion,
\item $\Sigma(c)$ is the class of processes $\sigma \geq 0$ on Wiener
  space which are progressively measurable with respect to the
  filtration generated by $W$ and such that
  \begin{equation*}
    a(\sigma_t) \leq c
  \end{equation*}
  for
  \begin{equation*}
    a(\sigma) \set
   \begin{cases}
\frac{1}{2} \frac{\sigmalow^2-\sigma^2}{\sigmalow}, & 0 \leq \sigma
\leq \sigmalow,\\
0, & \sigmalow \leq \sigma \leq \sigmahigh,\\
\frac{1}{2} \frac{\sigma^2-\sigmahigh^2}{\sigmahigh}, &
\sigmahigh \leq \sigma,
   \end{cases}
  \end{equation*}
\item $S^{\sigma}$ denotes the stock price evolution with
  volatility $\sigma$:
$$
S^{\sigma}_t \set s_0 \exp\left(\int_0^t \sigma_s \,dW_s - \frac{1}{2}
  \int_0^t \sigma^2_s \,ds\right), \quad 0 \leq t \leq 1.
$$
\end{itemize}
\end{thm}
The proof of this result is deferred to Section~\ref{sec:4}. One way
to interpret Theorem~\ref{thm2.3} is that it describes the scaling
limit of super-replication prices as a convex risk measure for the
payoff; see~\cite{FS} or~\cite{FG}.  The class $\Sigma(c)$
parametrizes the volatility models which one chooses to take into
account and the integral term measures the relevance one associates
with the payoff's mean under any such model. From this perspective the
most relevant models are those where the volatility stays within the
prescribed uncertainty region $[\sigmalow,\sigmahigh]$ (so that the
integral term vanishes). One also considers models with higher or
lower volatilities $\sigma_t$, though, (for as long as $a(\sigma_t)\leq c$),
but diminishes their relevance according to the average difference of
their local variances from $\sigmalow^2$ on the low and $\sigmahigh^2$
on the high side as measured by $a(\sigma_t)$. In particular, the
scaling limit of our super-replication prices is bounded from below
by the $G$--expectation
$\sup_{\sigma \in [\sigmalow,\sigmahigh]} \mathbb{E}^W \left[\mathbb
  F(S^{\sigma})\right].$
In the frictionless case (where $h\equiv0$, $\hat H\equiv\infty$) the
penalty for choosing a volatility model with values outside the
interval $[\sigmalow,\sigmahigh]$ is infinite and we recover the
well-known frictionless characterization of super-replication prices
under uncertainty as the payoff's $G$-expectation.


\section{Proofs}\label{sec:3}

In this section we carry out the proofs of Theorems~\ref{thm2.1}
and~\ref{thm2.3}.

\subsection{Proof of Theorem \ref{thm2.1}}
\label{sec:discretedualityproof}

Theorem~\ref{thm2.1} is concerned with the identity $V(\F) = U(\F)$
where
\begin{equation*}\label{eq:P1}
  U(\F) \set \sup_{\mathbb P\in \mathcal{P}_{\sigmalow,\sigmahigh}}
  \mathbb E_{\mathbb P}\left[\F(S)-\sum_{n=0}^{N-1}  G_n\left(\frac{\mathbb{E}_{\mathbb  P}[S_N\;|\;\mathcal{F}_n]-S_n}{S_n}\right)  \right].
\end{equation*}

As a first step we note:
\begin{lem}\label{lem:1}
 For any measurable $\F:\mathbb{R}_+^{N+1} \to \mathbb{R}_+$ we have
 that $V(\F) \geq U(\F)$.
\end{lem}
\begin{proof}
Let $\pi=(y,\gamma)$ super-replicate $\F(S)$. Then we have
\begin{align*}
\F(S) & \leq Y^\pi_N = y + \sum_{n=0}^{N-1}
\left( \gamma_n (S_{n+1}-S_n)
-g_k((\gamma_n-\gamma_{n-1})S_n)\right)\\
&=y + \sum_{n=0}^{N-1}
\left((\gamma_n-\gamma_{n-1}) (S_N-S_n) -g_n((\gamma_n-\gamma_{n-1})S_n)\right).
\end{align*}
Taking (conditional) expectations with respect to any $\mathbb{P} \in
\mathcal{P}_{\sigmalow,\sigmahigh}$ this shows that
\begin{align*}
\mathbb{E}_{\mathbb{P}} [\F(S)] &
 \leq
y + \mathbb{E}_{\mathbb{P}}\bigg[ \sum_{n=0}^{N-1}
 (\gamma_n-\gamma_{n-1})(\mathbb{E}_{\mathbb{P}}[S_N\;|\;\mathcal{F}_n]-S_n)-
g_n((\gamma_{n}-\gamma_{n-1})S_n)\bigg]\\
&=y + \mathbb{E}_{\mathbb{P}}\bigg[ \sum_{n=0}^{N-1}
 (\gamma_n-\gamma_{n-1})S_n\frac{\mathbb{E}_{\mathbb{P}}[S_N\;|\;\mathcal{F}_n]-S_n}{S_n}
-g_n((\gamma_{n}-\gamma_{n-1})S_n)\bigg]\\&\leq
y +\mathbb{E}_{\mathbb{P}}\left[ \sum_{n=0}^{N-1}
 G_n\left(\frac{\mathbb{E}_{\mathbb{P}}[S_N\;|\;\mathcal{F}_n]-S_n}{S_n}\right)\right]
\end{align*}
where the final estimate follows from the definition of the dual
functions $G_n$, $n=0,\dots,N$.

Since this holds for arbitrary $\mathbb{P}\in
\mathcal{P}_{\sigmalow,\sigmahigh}$ and any initial wealth $y$ for
which we can find a super-replicating strategy, the preceding
estimate yields $V(\F) \geq U(\F)$.
\end{proof}

We next observe that an identity analogous to $U(\F)=V(\F)$ holds for
multinomial models:
\begin{lem}\label{lem:2}
  For $k \in \{1,2,\dots\}$ consider the finite set
  \begin{equation*}
    \Omega^k \set
  \left\{x \in \Omega \;:\; |x_i| =\frac{j}{k} \sigmalow +
    \left(1-\frac{j}{k}\right)\sigmahigh \text{ for some }
      j \in \{0,\dots,k\}\right\}
  \end{equation*}
 and let
  $\mathcal{P}_{\sigmalow,\sigmahigh}^k$ be the subset of
  $\mathcal{P}_{\sigmalow,\sigmahigh}$ which contains those discrete
  probability measures that are supported by $\Omega^k$.

 Then we have
 \begin{equation*}
   V^k(\F) = U^k(\F)
 \end{equation*}
 where
 \begin{equation*}
   V^k(\F) = \inf\{y \;:\; Y^{(y,\gamma)}_N(\omega) \geq
   \F(S(\omega)) \, \forall \omega \in \Omega^k \text{ for some
     strategy } \gamma\}
 \end{equation*}
 and
 \begin{equation*}
     U^k(\F)=\sup_{\mathbb P\in \mathcal{P}^k_{\sigmalow,\sigmahigh}}
  \mathbb E_{\mathbb P}\left[\F(S)-\sum_{n=0}^{N-1}  G_n\left(\frac{\mathbb{E}_{\mathbb  P}[S_N\;|\;\mathcal{F}_n]-S_n}{S_n}\right)  \right].
 \end{equation*}
\end{lem}
\begin{proof}
  For $k=1$, i.e., in the binomial case, this is just Theorem~3.1 in
  \cite{DS1}. This result is proved by observing that the identity can
  be cast as a finite dimensional convex duality claim. The same
  reasoning actually applies to the multinomial setup with $k>1$ as
  well. This establishes our claim.
\end{proof}
In a third step we argue how to pass to the limit $k \uparrow \infty$,
first for continuous~$\F$:
\begin{lem}\label{lem:3}
With the notation of Lemma~\ref{lem:2} we have
\begin{equation}
  \label{eq:9new}
  U^k(\F) \leq U(\F), \quad k=1,2,\dots
\end{equation}
If $\F$ is continuous we have furthermore
\begin{equation}
  \label{eq:10new}
  \liminf_{k \uparrow \infty} V^k(\F) \geq V(\F).
\end{equation}
\end{lem}
\begin{proof}
  Estimate~\eqref{eq:9new} is immediate from the definitions of
  $U^k(\F)$ and $U(\F)$ as $\mathcal{P}_{\sigmalow,\sigmahigh}^k
  \subset \mathcal{P}_{\sigmalow,\sigmahigh}$.

  To prove~\eqref{eq:10new} take, for $k=1,2,\dots$, a strategy
  $\tilde{\gamma}^k$ such that $Y^{V^k(\F)+1/k,\tilde{\gamma}^k}_N
  \geq \F(S)$ on $\Omega^k$.  We will show below that without loss of
  generality we can assume that the sequence of $\tilde{\gamma}^k$s is
  uniformly bounded, i.e.,
  \begin{equation}
    \label{eq:11}
  |\tilde{\gamma}^k| \leq \mathcal C \text{ for some
  constant } \mathcal C>0.
  \end{equation}
  Now consider the strategies $\gamma^k \set \tilde{\gamma}^k \circ
  p^k$ where $p^k:\Omega \to \Omega^k$ is the projection which maps
  $\omega=(x_1,\dots,x_N)$ to
  $p^k(\omega)=\tilde{\omega} \set (\tilde{x}_1,\dots,\tilde{x}_N) \in \Omega^k$ with
  \begin{equation*}
   \tilde{x}_i \set \max\left\{x \leq x_i \;:\; |x|=\frac{j}{k} \sigmalow +
    \left(1-\frac{j}{k}\right)\sigmahigh \text{ for some }
      j \in \{0,\dots,k\}\right\}.
  \end{equation*}
  For any initial capital $y$, we get
  \begin{align*}
    &Y^{y,\gamma^k}_N (\omega)-Y^{y,\tilde{\gamma}^k}_N (\tilde{\omega})=\\
    &\sum_{n=0}^{N-1}
    \tilde{\gamma}^k_n(\tilde{\omega})\left((S_{n+1}(\omega)-S_n(\omega))
      -(S_{n+1}(\tilde{\omega})-S_n(\tilde{\omega}))\right)\\
&- \sum_{n=0}^{N-1}
\left(g_n(\omega,(\tilde{\gamma}^k_n(\tilde{\omega})-\tilde{\gamma}^k_{n-1}(\tilde{\omega}))S_n(\omega))-
g_n(\tilde{\omega},(\tilde{\gamma}^k_n(\tilde{\omega})-\tilde{\gamma}^k_{n-1}(\tilde{\omega}))S_n(\tilde{\omega}))\right).
  \end{align*}
  Because $|\tilde{\gamma}^k|\leq \mathcal C$ uniformly in $k=1,2,\dots$, the
  first of these two sums has absolute value less than
  \begin{align*}
         2\mathcal C \sum_{n=0}^{N-1} w(S_n,|\omega-\tilde{\omega}|)
  \end{align*}
  where for any function $f$ we let $w(f,\delta)$, $\delta >0$, denote
  the modulus of continuity over its domain. Similarly, we get for the
  second sum that its absolute value does not exceed
  \begin{equation*}
    \sum_{n=0}^{N-1} w(g_n \mbox{}_{|_{\Omega \times [-2\mathcal C s_0e^{\sigmahigh n},2\mathcal C s_0e^{\sigmahigh n}]}} ,|\omega-\tilde{\omega}|+2\mathcal C
    w(S_n,|\omega-\tilde{\omega}|)).
  \end{equation*}
  By continuity of $S$ and $g_n$, $n=0,\dots,N-1$, both of these
  bounds tend to 0 as $|\omega-\tilde{\omega}| \to 0$. It follows that
  there are $\epsilon_k \downarrow 0$ as $k \uparrow \infty$ such that
  $$
    Y^{y,\tilde{\gamma}^k}_N(\tilde{\omega}) \leq
    Y^{y,\gamma^k}_N(\omega)+\epsilon_k \text{ for all }
    |\omega-\tilde{\omega}| \leq 1/k, y \in \mathbb{R}.
  $$
  By assumption $\F$ is also continuous and so we get
  \begin{align*}
    \F(S)(\omega)& \leq\F(S)(\tilde{\omega})+w(\F(S),|\omega-\tilde{\omega}|) \\
    & \leq Y^{V^k(\F)+1/k,\tilde{\gamma}^k}_N(\tilde{\omega})+w(\F(S),|\omega-\tilde{\omega}|)\\
    & \leq Y^{V^k(\F)+1/k,\gamma^k}_N(\omega)+\epsilon_k +w(\F(S),1/k)
  \end{align*}
  It follows that $V(\F) \leq V^k(\F)+1/k+\epsilon_k+w(\F(S),1/k)$ which implies our
  claim~\eqref{eq:10new}.

  It remains to prove the uniform boundedness~\eqref{eq:11} of the
  sequence $(\tilde{\gamma}^k)_{k=1,2,\dots}$. Clearly, $y^k \set
  V^k(\F)+1/k \leq A$, $k=1,2,\dots$, for some $A>0$.  For any
  $\tilde{\pi}^k=(y^k,\tilde{\gamma}^k)$, $k=1,2,\dots$, we
  will prove by induction over $n$ that
  \begin{equation}\label{4.3}
    Y^{\tilde{\pi}^k}_n(\tilde{\omega})\leq
    A\left(1+e^{\sigmahigh}\right)^{n}
    \text{ and }
    |\tilde{\gamma}^k_n(\tilde{\omega})|\leq \frac{A\left(1+e^{\sigmahigh}\right)^{n}}{(1-e^{-\sigmahigh})S_n(\tilde{\omega})}
  \end{equation}
  for any $\tilde{\omega}=(\tilde{x}_1,\dots,\tilde{x}_N)\in \Omega^k$
  and $n=0,1,...,N$. Since $S_n\geq s_0e^{-\sigmahigh N}$, our
  claim~(3.4) then holds for $\mathcal C \set
  {A\left(1+e^{\sigmahigh}\right)^{N}}/((1-e^{-\sigmahigh})s_0e^{-\sigmahigh
    N})$.

  Since each $\tilde{\pi}^k$ super-replicates a positive
  claim,  we must have $Y^{\tilde{\pi}^k}_1 \geq 0$ in any
  possible scenario. In particular, we have
  $$
  y^k+\tilde{\gamma}^k_0 s_0  (e^{\sigmahigh}-1)\geq 0 \text{ and }
  y^k+\tilde{\gamma}^k_0 s_0  (e^{-\sigmahigh}-1)\geq 0
 $$
 which allows us to conclude that $|\tilde{\gamma}^k_0|\leq
 \frac{A}{s_0(1-e^{-\sigmahigh})}$. Thus~\eqref{4.3} holds for $n=0$.
 Next, assume that~(\ref{4.3}) holds for $n$ and let us prove it for
 $n+1$.  From the induction assumption we get
\begin{align*}
Y^{\tilde{\pi}^k}_{n+1}(\tilde{\omega})
&\leq Y^{\tilde{\pi}^k}_{n}(\tilde{\omega})+
\tilde{\gamma}^k_n(\tilde{\omega})(S_{n+1} (\tilde{\omega})-
S_n(\tilde{\omega}))\\
&\leq A\left(1+e^{\sigmahigh}\right)^{n}+
\frac{A\left(1+e^{\sigmahigh}\right)^{n}}{(1-e^{-\sigmahigh})S_n(\tilde{\omega})}S_n(\tilde{\omega})
(e^{\sigmahigh}-1)
= A\left(1+e^{\sigmahigh}\right)^{n+1},
\end{align*}
as required. Again, the portfolio valued at time $n+2$ should be non
negative, for any possible scenario.  Thus,
$$
Y^{\tilde{\pi}^k}_{n+1}(\tilde{\omega})+\tilde{\gamma}^k_{n+1}(\tilde{\omega})S_{n+1}(\tilde{\omega})
(e^{\sigmahigh}-1)\geq 0$$
and
$$Y^{\tilde{\pi}^k}_{n+1}(\tilde{\omega})+\tilde{\gamma}^k_{n+1}(\tilde{\omega})
S_{n+1}(\tilde{\omega}) (e^{-\sigmahigh}-1)\geq 0
$$
 and so,
 $$
|\tilde{\gamma}^k_{n+1}(\tilde{\omega})|
\leq \frac{Y^{\tilde{\pi}^k}_{n+1}(\tilde{\omega})}{(1-e^{-\sigmahigh})S_{n+1}(\tilde{\omega})}
\leq
\frac{A\left(1+e^{\sigmahigh}\right)^{n+1}}{(1-e^{-\sigmahigh})S_{n+1}(\tilde{\omega})}.
$$
This completes the proof of~\eqref{4.3}.
\end{proof}

It is immediate from Lemmas~\ref{lem:1}--Lemma~\ref{lem:3} that
$V(\F)=U(\F)$ for continuous functions $\F$. For upper-semicontinuous
$\F$ we can find continuous functions $\F^k$ with $\sup_\Omega \F(S)
\geq \F^k(S) \geq \F(S)$ such that
$$\limsup_{k \uparrow \infty,
  \omega_k \to \omega} \F^k(S(\omega_k)) \leq \F(S(\omega))$$ for any
$\omega \in \Omega$; see, e.g., Lemma~5.3 in \cite{DS}.  The proof of
Theorem~\ref{thm2.1} will thus follow from the series of inequalities
\begin{equation*}
  V(\F) \geq U(\F) \geq \limsup_{k \uparrow \infty} U(\F^k) = \limsup_{k \uparrow \infty} V(\F^k) \geq V(\F)
\end{equation*}
where the first inequality is due to Lemma~\ref{lem:1}, the last holds
because $\F^k \geq \F$ and the identity follows because our claim is
already established for continuous $\F^k$. Hence, the only estimate
still to be shown is the second one:
\begin{lem}\label{lem:4}
  If $\F$ is approximated by $\F^k$, $k=1,2\dots$, as above we have
\begin{equation*}\label{eq:12}
U(\F) \geq \limsup_{k \uparrow \infty} U(\F^k).
\end{equation*}
\end{lem}
\begin{proof}
 Without loss of generality we can assume that $(U(\F^k))_{k=1,2,\dots}$ converges in
 $\mathbb{R}$. By definition of $U(\F^k)$ there is $\mathbb{P}^k \in
 \cP_{\sigmalow,\sigmahigh}$ such that, for $k=1,2,\dots$,
 \begin{align}\label{eq:13}
   U(\F^k) -\frac1k\leq  \mathbb{E}_{\mathbb{P}^k}\left[\F^k(S)-\sum_{n=0}^{N-1}
     G_n\left(\frac{\mathbb{E}_{\mathbb{P}^k}[S_N\;|\;\mathcal{F}_n]-S_n}{S_n}\right)
   \right].
 \end{align}
 We wish to show that the $\limsup$ of the right side of~\eqref{eq:13}
 as $k \uparrow \infty$ is not larger than $U(\F)$.  To this end,
 denote by $\Pi$ the set of Borel probability measures on
 $\Omega\times [0,s_0e^{\sigmahigh N}]^N$.  Since $\Omega\times
 [0,s_0e^{\sigmahigh N}]^N$ is compact, so is $\Pi$ when endowed with
 the weak topology. Now consider the sequence of probabilities
 measures in $\Pi$ obtained by considering the law of
 \begin{align*}
   Z^k\set(X,\mathbb{E}_{\mathbb{P}^k}[S_N],\mathbb{E}_{\mathbb{P}^k}[S_N\;|\;\mathcal{F}_1],...,\mathbb{E}_{\mathbb{P}^k}[S_N\;|\;\mathcal{F}_{N-1}])
 \end{align*}
 under $\mathbb{P}^k$ for $k=1,2,\dots$. Due to Prohorov's theorem, by
 possibly passing to a subsequence, we can assume without loss of
 generality that this sequence converges weakly. By Skorohod's
 representation theorem there thus exists a probability space
 $(\hat{\Omega},\hat{\mathcal{F}},\hat{\mathbb{P}})$ with a
 $\hat{\mathbb{P}}$-almost surely convergent sequence of random
 variables $\hat{Z}^k$, $k=1,2,\dots$, taking values in $\Omega \times
 [0,s_0e^{\sigmahigh N}]^N$, whose laws under $\hat{\mathbb{P}}$ coincide, respectively,
 with those of $Z^k$ under $\mathbb{P}^k$, $k=1,2,\dots$.
 Let $\hat{Z}^\infty$ denote the $\hat{\mathbb{P}}$-a.s. existing limit of $\hat{Z}^k$,
 $k=1,2,\dots$, and write it as
 \begin{align*}
  \hat{Z}^\infty=(\hat{X}^\infty,Y_0,...,Y_{N-1}).
 \end{align*}
 We will show that
 \begin{align}\label{eq:14}
 \mathbb{E}_{\hat{\mathbb{P}}}[ S_N(\hat X^{\infty})|\hat{X}^\infty_1,...,\hat{X}^{\infty}_n]=
 \mathbb{E}_{\hat{\mathbb{P}}}[Y_n|\hat{X}^\infty_1,...,\hat{X}^{\infty}_n], \quad n=0,\dots, N.
 \end{align}
 By construction of $\hat{Z}^k$ we have for the right side
 of~\eqref{eq:13}:
 \begin{align*}
   \mathbb{E}_{\mathbb{P}^k}&\left[\F^k(S)-\sum_{n=0}^{N-1}
     G_n\left(X,\frac{\mathbb{E}_{\mathbb{P}^k}[S_N\;|\;\mathcal{F}_n]-S_n}{S_n}\right)
   \right]\\
&=\mathbb{E}_{\hat{\mathbb{P}}}\left[\F^k(S(\hat{X}^k))-\sum_{n=0}^{N-1}
     G_n\left(\hat{X}^k,\frac{\mathbb{E}_{\hat{\mathbb{P}}}[S_N(\hat{X}^k)|\hat{X}^k_1,...,\hat{X}^k_n]-S_n(\hat{X}^k)}{S_n(\hat{X}^k)}\right)
   \right].
 \end{align*}
 The $\hat{\mathbb{P}}$-a.s. convergence of $\hat{Z}^k$ and the
 construction of the sequence of $\F^k$ imply that the $\limsup_{k
   \uparrow \infty}$ of the term inside this last expectation is
 $\hat{\mathbb{P}}$-a.s. not larger than
\begin{eqnarray*}
\F(S(\hat{X}^\infty))-\sum_{n=0}^{N-1}
   G_n\left(\hat{X}^\infty,\frac{Y_n-S_n(\hat{X}^\infty)}{S_n(\hat{X}^\infty)}\right)
 \end{eqnarray*}
 where we used the lower semi-continuity of $G_n$.  Because of the
 boundedness of $\F$ on compact sets and because $G \geq 0$, it then follows by
 Fatou's lemma that the $\limsup$ of the right side of~\eqref{eq:13}
 is not larger than
 \begin{align*}
   \mathbb{E}_{\hat{\mathbb{P}}}\left[\F(S(\hat{X}^\infty))-\sum_{n=0}^{N-1}
   G_n\left(\hat{X}^\infty,\frac{Y_n-S_n(\hat{X}^\infty)}{S_n(\hat{X}^\infty)}\right)\right].
   \end{align*}
 From the definitions it follows that $G_n(\omega,\alpha)$ is adapted
 and $G_n(\omega,\cdot)$ is convex. This together with
 the Jensen inequality and (\ref{eq:14}) yields that for any $n<N$
 \begin{eqnarray*}
&\mathbb{E}_{\hat{\mathbb{P}}}\left[G_n\left(\hat{X}^\infty,\frac{Y_n-S_n(\hat{X}^\infty)}{S_n(\hat{X}^\infty)}\right)\bigg|\hat{X}^{\infty}_1,...,\hat{X}^{\infty}_n\right]\geq\\
&G_n\left(\hat{X}^\infty,\frac{ \mathbb{E}_{\hat{\mathbb{P}}}[S_N(\hat{X}^\infty)|\hat{X}^{\infty}_1,...,\hat{X}^{\infty}_n]-S_n(\hat{X}^\infty)}{S_n(\hat{X}^\infty)}\right).
 \end{eqnarray*}
We conclude that  the $\limsup$ of the right side of~\eqref{eq:13}
 is not larger than
 \begin{align*}
  \mathbb{E}_{\hat{\mathbb{P}}}\left[\F(S(\hat{X}^\infty))-\sum_{n=0}^{N-1}
   G_n\left(\hat{X}^\infty,\frac{\mathbb{E}_{\hat{\mathbb{P}}}[S_N(\hat{X}^\infty)|\hat{X}^{\infty}_1,...,\hat{X}^{\infty}_n]-S_n(\hat{X}^\infty)}{S_n(\hat{X}^\infty)}\right)\right].
 \end{align*}
 Since the distribution of $\hat{X}^\infty$ is an element in
 $ \mathcal{P}_{\sigmalow,\sigmahigh}$, this last expectation is not
 larger than $U(\F)$ as we had to show.

 It remains to establish~\eqref{eq:14}.  Let $n<N$ and let
 $f:\mathbb{R}^{n}\rightarrow\mathbb{R}$ be a continuous bounded
 function.  From the dominated convergence theorem it follows that
\begin{align*}
 \mathbb{E}_{\hat{\mathbb{P}}}[Y_n f(\hat X^{\infty}_1,...,\hat X^{\infty}_n)]=
&\lim_{k\rightarrow\infty}
 \mathbb{E}_{\hat{\mathbb{P}}}\left[\mathbb{E}_{\hat{\mathbb{P}}}[S_N(\hat X^{k})|\hat{X}^{k}_1,...,\hat{X}^{k}_n]
 f(\hat X^{k}_1,...,\hat X^{k}_n)\right]\\=
&\lim_{k\rightarrow\infty} \mathbb{E}_{\hat{\mathbb{P}}}[S_N(\hat X^{k})
 f(\hat X^{k}_1,...,\hat X^{k}_n)]\\=&
 \mathbb{E}_{\hat{\mathbb{P}}}[S_N(\hat X^{\infty}) f(\hat X^{\infty}_1,...,\hat X^{\infty}_n)].
\end{align*}
Thus by applying standard density arguments
we obtain (\ref{eq:14}). This accomplishes our proof.
\end{proof}


\subsection{Proof of Theorem \ref{thm2.3}}
\label{sec:4}

For the proof of the asserted limit
\begin{equation}
  \label{eq:15}
  \lim_{N\rightarrow\infty}V^{\sigmalow/\sqrt{N},\sigmahigh/\sqrt{N}}_{g^{N,c}}(\F)
  = \sup_{\sigma \in \Sigma(c)} \mathbb{E}^W
  \left[\mathbb F(S^{\sigma})-\int_{0}^1\hat H_t(S^{\sigma})a^2(\sigma_t) dt\right]
\end{equation}
we first have to go through some technical preparations in
Section~\ref{sec:technical} before we can establish `$\leq$' and then
`$\geq$' in~\eqref{eq:15} in Sections~\ref{sec:leq} and~\ref{sec:geq},
respectively.

\subsubsection{Technical preparations}\label{sec:technical}

Let us start by recalling that, for $N=1,2,\dots$,
\begin{equation*}
  \Omega^N = \{\omega^N=(x_1,\dots,x_N) \;:\; \sigmalow/\sqrt{N}
  \leq |x_n| \leq \sigmahigh/\sqrt{N}, n=1,\dots,N\}
\end{equation*}
allows for the definition of the canonical process
\begin{equation*}
 X^N_n(\omega^N)=x_n \text{ for } n=1,\dots,N,
 \;\omega^N=(x_1,\dots,x_N) \in \Omega^N.
\end{equation*}
We thus can consider the canonical filtration
\begin{equation*}
  \mathcal{F}^N_n \set\sigma(X^N_m,m=1,\dots,n), \quad n=0,\dots,N,
\end{equation*}
which clearly is the same as the one generated by
\begin{equation*}
 S^N_n = s_0 \exp\left( \sum_{m=1}^n X_m\right), \quad n=0,\dots,N,
\end{equation*}
since
\begin{equation*}
 X^N_n = \ln S^N_n - \ln S^N_{n-1}, \quad n=1,\dots,N.
\end{equation*}

It will be convenient to let, for a vector $y = (y_0,\dots,y_N) \in
\mathbb{R}^{N+1}$, the function $\overline{y} \in C[0,1]$ denote the
continuous linear interpolation on $[0,1]$ determined by
$\overline{y}_{n/N} = y_n$, $n=0,\dots,N$. This gives us, in particular,
the continuous time analog $(\overline{S}^N_t)_{0 \leq t\leq 1}$ of
$(S^N_n)_{n=0,\dots,N}$.

Our first observation is that the continuity of $\mathbb F$ allows us
to write the supremum in~\eqref{eq:15} in different ways:

\begin{lem}\label{lem:101}
  Let
  \begin{align*}
   R \set  \sup_{\sigma \in \Sigma(c)} \mathbb{E}^W
  \left[\mathbb F(S^{\sigma})-\int_{0}^1\hat H_t(S^{\sigma})a^2(\sigma_t) dt\right]
  \end{align*}
 denote the right side of~\eqref{eq:15}.

 \begin{itemize}
 \item[(i)] We have
   \begin{align}\label{eq:16}
     R = \sup_{\mathbb{P} \in \mathcal{P}_{\sigmalow,\sigmahigh,c}}
     \mathbb{E}_{\mathbb{P}}\left[\mathbb F(S)-\int_{0}^1\hat H_t(S)a^2\left(\sqrt{\frac{d\langle S \rangle_t}{dt}/S_t^2 }\right) dt\right]
   \end{align}
   where $\mathcal{P}_{\sigmalow,\sigmahigh,c}$ denotes the class of
   probabilities $\mathbb{P}$ on $(C[0,1],\mathcal{B}(C[0,1]))$ under
   which the coordinate process $S=(S_t)_{0 \leq t \leq 1}$ is a
   strictly positive martingale starting at $S_0=s_0$ whose quadratic
   variation is absolutely continuous with
   \begin{align*}
     a\left(\sqrt{\frac{d\langle S \rangle_t}{dt}/S_t^2 }\right) \leq
     c, \quad 0 \leq t \leq 1, \; \mathbb{P}\text{-almost surely}.
   \end{align*}
 \item[(ii)] The supremum defining $R$ does not change when we take it
   over $\tilde{\Sigma}(c) \subset \Sigma(c)$, the class of
   progressively measurable processes $\tilde{\sigma}: [0,1] \times
   C[0,1] \to \mathbb{R}_+ $ on the Wiener space
   $(C[0,1],\mathcal{B}(C[0,1]),\mathbb{P}^W)$ such that
  \begin{itemize}
  \item There is $\delta>0$ such that
    \begin{equation}\label{3.3+}
      \sigmalow(\sigmalow-2c)^{+}+\delta\leq \tilde{\sigma}^2\leq \sigmahigh(\sigmahigh+2c)-\delta
    \end{equation}
    uniformly on $[0,1] \times C[0,1]$ and such that, in addition,
    \begin{equation}
      \label{3.3++}
      \tilde{\sigma} \equiv \sigmalow, \quad \text{ on } [1-\delta,1] \times C[0,1].
    \end{equation}
  \item $\tilde{\sigma}$ is Lipschitz continuous on $[0,1] \times
    C[0,1]$.
  \end{itemize}
 \end{itemize}
\end{lem}
\begin{proof}
The proof is done similarly to the proof of Lemmas 7.1--7.2 in \cite{DS1}.

\end{proof}

The following technical key lemma can be viewed as an adaption of
Kusuoka's results from~\cite{K} on super-replication with proportional
transaction costs to our uncertain volatility setting with nonlinear
costs:

\begin{lem}\label{lem:102}
   Under the assumptions of Theorem~\ref{thm2.3} the following holds true:
  \begin{itemize}
  \item[(i)] Let $c>0$ and, for $N=1,2,\dots$, let $\mathbb{Q}^N$ be a
    probability measure on $(\Omega^N,\mathcal{F}^N_N)$ for which
    \begin{equation}
      \label{eq:17}
        M^N_n \set \mathbb{E}^{\mathbb{Q}^N}[S^N_N \;|\;
        \mathcal{F}^N_n], \quad n=0,\dots,N,
    \end{equation}
    is close to $S^N$ in the sense that $\mathbb{Q}^N$-almost surely
    \begin{equation}
      \label{eq:18}
       \left|\frac{M^N_n-S^N_n}{S^N_n}\right| \leq \frac{c}{\sqrt{N}},
       \quad n=0,\dots,N.
    \end{equation}
    Then we have
    \begin{equation}
      \label{eq:19}
       \sup_{N=1,2,\dots}
       \mathbb{E}_{\mathbb{Q}^N}\left[\left(\max_{n=0,\dots,N}
           S^N_n\right)^p\right]<\infty \text{ for any } p>0
    \end{equation}
    and, with
    \begin{equation}
      \label{eq:20}
        Q^N_n \set \sum_{m=1}^n (X^N_m)^2 + 2\sum_{m=1}^n
        \frac{M^N_m-S^N_m}{S^N_m} X^N_m, \quad n=0,\dots,N,
    \end{equation}
    there is a subsequence, again denoted by $N$, such that, for $N
    \uparrow \infty$,
    \begin{equation}
      \label{eq:21}
      \law(\overline{S}^N,\overline{M}^N,\overline{Q}^N\, |\, \mathbb{Q}^N )
      \Rightarrow \law\left(\left.S,S,\int_0^. \frac{d\langle S\rangle_s}{S^2_s}
      \,\right|\,\mathbb{P}\right)
    \end{equation}
    where $\mathbb{P}$ is a probability measure in
    $\mathcal{P}_{\sigmalow,\sigmahigh,c}$, $S$ is as in
    Lemma~\ref{lem:101}~(i), and where, as before, $\overline{S}^N$ etc. denote
    the continuous interpolations on $[0,1]$ induced by the vector
    $S^N=(S^N_n)_{n=0,\dots,N}$ etc.
\item[(ii)] For any $c > 0$ and $\tilde{\sigma} \in
    \tilde{\Sigma}(c)$ as in Lemma~\ref{lem:101}~(ii), there exists a
    sequence of probability measures $\mathbb{Q}^N$, $N=1,2,\dots$, as
    in (i) such that the weak convergence in~\eqref{eq:21} holds with
    $\mathbb{P} \set \law(S^{\tilde{\sigma}} \,|\, \mathbb{P}^W)$. In
    addition we get the weak convergence (as $N\uparrow\infty$)
    \begin{equation}
      \label{eq:22}
      \law\left(\left.\overline{S}^N,\overline{M}^N, \overline{\sqrt{N}\frac{|M^N-S^N|}{S^N}}\, \right|\, \mathbb{Q}^N \right)
      \Rightarrow \law(S^{\tilde{\sigma}},S^{\tilde{\sigma}},a(\tilde{\sigma})
      \,|\,\mathbb{P}^W) .
    \end{equation}
 \end{itemize}
\end{lem}
\begin{proof}

  Let us first focus on claim~(i). It obviously suffices to
  prove~\eqref{eq:19} only for $p \in \{1,2,\dots\}$. For this we
  proceed similarly as Kusuoka for his claim~(4.23) in~\cite{K} and write
    \begin{align*}
      (M^N_{n+1})^{p} & =
      (M^N_{n})^{p}\left(1+\frac{M^N_{n+1}-M^N_n}{M^N_n}\right)^{p}\\
     &=  (M^N_{n})^{p} \left(1+\sum_{j=1}^{p} \binom{p}{j}
       \left(\frac{M^N_{n+1}-M^N_n}{M^N_n}\right)^{j}\right).
    \end{align*}
    Now observe that, when taking the $\mathbb{Q}^N$-expectation, the
    contribution from the summand for $j=1$ can be dropped since it has vanishing
     $\mathcal{F}^N_n$-conditional expectation due to the martingale
     property of $M^N$ under $\mathbb{Q}^N$. From~\eqref{eq:18} and
     $S^N_{n+1}/S^N_n = \exp(X^N_n) \in
     [e^{-\sigmahigh/\sqrt{N}},e^{\sigmahigh/\sqrt{N}}]$, it follows
     that $M^N_{n+1}/M^N_n=1+O(1/\sqrt{N})$ where the random
     $O(1/\sqrt{N})$-term becomes small uniformly in $n$ and
     $\omega^N$. Therefore the summands for $j=2,\dots,p$ are
     uniformly of the order $O(1/N)$. Thus, we obtain
     \begin{align*}
       \mathbb{E}_{\mathbb{Q}^N}\left[\left(M^N_{n+1}\right)^{p}\right]
       = (1+O(1/N))
       \mathbb{E}_{\mathbb{Q}^N}\left[\left(M^N_{n}\right)^{p}\right]
     \end{align*}
     and, so upon iteration,
     \begin{align*}
       \mathbb{E}_{\mathbb{Q}^N}\left[\left(M^N_{N}\right)^{p}\right]
       = (1+O(1/N))^N\left(M^N_{0}\right)^{p}.
     \end{align*}
     Clearly $(1+O(1/N))^N$ is bounded in $N$. The same holds for
     $M^N_0 = s_0(1+O(1/\sqrt{N})$, where we used~\eqref{eq:18}) and
     $S^N_0=s_0$. Hence,
     \begin{align*}
           \sup_{N=1,2,\dots}  \mathbb{E}_{\mathbb{Q}^N}\left[\left(M^N_{N}\right)^{p}\right]<\infty
     \end{align*}
     which by Doob's inequality entails that even
     \begin{align}\label{eq:23}
           \sup_{N=1,2,\dots}  \mathbb{E}_{\mathbb{Q}^N}\left[\left( \max_{n=0,\dots,N}M^N_{n}\right)^{p}\right]<\infty.
     \end{align}
     Because of~\eqref{eq:18}, this yields our claim~\eqref{eq:19}.

     Let us next focus on $\law(\overline{M}^N \;|\;
     \mathbb{Q}^N)_{N=1,2,\dots}$ for which we will verify
     Kolmogorov's tightness criterion (see \cite{C}) on $C[0,1]$. To this end, recall
     that $M^N_{n+j}/M^N_{n+j-1}-1=O(1/\sqrt{N})$ and so the quadratic
     variation of $M^N$ satisfies
    \begin{align*}
      \langle M^{N}\rangle_{n+l}-\langle M^{N}\rangle_{n}&=
      \sum_{j=1}^{l} \mathbb{E}_{\mathbb{Q}^N}[(M^{N}_{n+j}-M^{N}_{n+j-1})^2\;|\;\mathcal{F}^{N}_{n+j-1}] \\
      &= \left(\max_{0\leq n\leq N} M^{N}_n\right)^2 O(l/N).
    \end{align*}
    From the Burkholder--Davis--Gundy inequality and the
    bound~\eqref{eq:23} we thus get
    \begin{align*}
      \mathbb{E}_{\mathbb{Q}^N}[(M^{N}_{n+l}-M^{N}_{n})^4]=O((l/N)^2)
    \end{align*}
    which readily gives
    Kolmogorov's criterion for our continuous interpolations
    $\overline{M}^N$, $N=1,2,\dots$.

    Having established tightness, we can find a subsequence, again
    denoted by $N$, such that $\law(\overline{M}^N \;|\;
    \mathbb{Q}^N)_{N=1,2,\dots}$  converges to the law of a continuous process
    $M$ on a suitable probability space
    $(\hat{\Omega},\hat{\mathcal{F}},\hat{\mathbb{P}})$. We will show
    next that this process $M$ is a strictly positive martingale. In
    fact, by Skorohod's representation theorem, we can assume that
    there are processes $\hat{M}^N$, $N=1,2,\dots$, on
    $(\hat{\Omega},\hat{\mathcal{F}},\hat{\mathbb{P}})$ with
   \begin{align*}
     \law(\overline{M}^N \;|\; \mathbb{Q}^N) = \law(\hat{M}^N \;|\; \hat{\mathbb{P}})
   \end{align*}
   which converge $\hat{\mathbb{P}}$-almost surely to $M$ as
   $N\uparrow \infty$. It is then immediate from~\eqref{eq:23} that
   the martingale property of $M^N$ under $\mathbb{Q}^N$ gives the
   martingale property of $\hat{M}$ under $\hat{\mathbb{P}}$. To see
   that $M$ is strictly positive we follow Kusuoka's argument
   for~(4.24) and (4.25) in his paper~\cite{K} and establish
   \begin{align}\label{eq:24}
     \sup_{N=1,2,\dots} \mathbb{E}_{\mathbb{Q}^N}\left[\max_{n=0,\dots,N}\left| \ln
         M^N_{n}\right|^2\right]<\infty
   \end{align}
   since this entails the $\hat{\mathbb{P}}$-integrability of $\max_{0 \leq t
     \leq 1} |\ln M_t|$ by Fatou's lemma.  For~\eqref{eq:24} recall
   that $M^N_{m}/M^N_{m-1} = 1+O(1/\sqrt{N})$ uniformly in $m$ and
   $\omega^N$. This allows us to use Taylor's
   expansion to obtain
   \begin{align*}
     \ln M^N_{m}-\ln M^N_{m-1} = \frac{M^N_{m}-M^N_{m-1}}{M^N_{m-1}}+ O(1/N).
   \end{align*}
   Upon summation over $m=1,\dots,n$, this gives in conjunction with
   $M^N_0=s_0(1+O(1/\sqrt{N}))$:
   \begin{align*}
     \max_{n =0,\dots,N} |\ln M^N_n| \leq \max_{n=0,\dots,N}
     \left|\sum_{m=1}^n \frac{M^N_{m}-M^N_{m-1}}{M^N_{m-1}}\right|+O(1).
   \end{align*}
   Hence
   \begin{align*}
     \mathbb{E}_{\mathbb{Q}^N}\left[\max_{n=0,\dots,N}\left| \ln
         M^N_{n}\right|^2\right] &\leq
     2\mathbb{E}_{\mathbb{Q}^N}\left[\max_{n=0,\dots,N}
       \left|\sum_{m=1}^n\frac{M^N_{m}-M^N_{m-1}}{M^N_{m-1}}\right|^2\right]
     +O(1)\\
     &\leq 8\mathbb{E}_{\mathbb{Q}^N}\left[
       \sum_{m=1}^N\left(\frac{M^N_{m}-M^N_{m-1}}{M^N_{m-1}}\right )^2\right]
     +O(1),
   \end{align*}
   where for the second estimate we used Doob's inequality for the
   martingale given by the sum for which we take the maximum in the
   above expression. Recalling again that
   $M^N_{m}/M^N_{m-1}-1=O(1/\sqrt{N})$ uniformly in $m$ and $\omega^N$
   the above expectation is of order $O(1)$ and we
   obtain~\eqref{eq:24}.

   Let us finally turn to the weak convergence~\eqref{eq:19} and
   introduce, for $N=1,2,\dots$, the auxiliary discrete stochastic integrals
   \begin{align*}
     Y^N_n \set \sum_{m=1}^{n} \frac{M^N_m-M^N_{m-1}}{S^N_{m-1}},
     \quad n= 0,\dots,N.
   \end{align*}
   By applying Theorem~4.3 in~\cite{DP}, \eqref{eq:18}--\eqref{eq:19} and
   the already established weak convergence of $\law(\overline{M}^N \;|\;
   \mathbb{Q}^N)_{N=1,2,\dots}$ on $C[0,1]$, we deduce the
   weak convergence
   \begin{align}\label{eq:25}
      \law&\left(\left. \left( \frac{1}{S^N_{[Nt]}},
              M^N_{[Nt]},Y^N_{[Nt]}\right)_{0 \leq t \leq
              1}\;\right|\; \mathbb{Q}^N\right)\\\nonumber & \Rightarrow\quad
      \law\left(\left. \left( \frac{1}{M_t},
              M_t,Y_t\right)_{0 \leq t \leq1}\;\right|\; \hat{\mathbb{P}}\right)\quad \text{ as }
            N\uparrow \infty
   \end{align}
   on the Skorohod space $D[0,1] \times D[0,1] \times D[0,1]$ where $Y
   \set \int_0^. dM_s/M_s$. Hence, $M = M_0
   \exp(Y-\langle{Y}\rangle/2)$ and, in particular, $\langle \ln M
   \rangle = \langle Y \rangle$. Moreover, again by Skorohod's
   representation theorem, we can assume that there are processes
   $1/\hat{S}^N$, $\hat{M}^N$ and $\hat{Y}^N$ on $(\hat{\Omega},
   \hat{F}, \hat{\mathbb{P}})$ whose joint law under
   $\hat{\mathbb{P}}$ coincides with that in~\eqref{eq:25} and which
   converge $\hat{\mathbb{P}}$-almost surely to $1/M$, $M$, and $Y$,
   respectively. Now, recalling that $|X^N_m| \leq
   \sigmahigh/\sqrt{N}$ we can Taylor expand $e^x=1+x+x^2/2+O(x^3)$
   so that with~\eqref{eq:18} we can write
   \begin{align*}
     \frac{M^N_m-M^N_{m-1}}{S^N_{m-1}} &=
      \left(1+\frac{M^N_m-S^N_{m}}{S^N_{m}}\right)e^{X^N_m}- \frac{M^N_{m-1}}{S^N_{m-1}}\\
     &=  \frac{M^N_m}{S^N_m}-\frac{M^N_{m-1}}{S^N_{m-1}}+ X^N_m+\frac{1}{2} (X^N_m)^2+
     \frac{M^N_m-S^N_{m}}{S^N_{m}}X^N_m + O(1/\sqrt{N}^3).
   \end{align*}
   Thus, with $Q^N$  as defined in~\eqref{eq:20}, we obtain
   upon summing over $m=1, \dots, n$:
   \begin{align*}
     Y^N_n = \frac{M^N_n}{S^N_n}-\frac{M^N_{0}}{S^N_{0}}+\ln S^N_n -
     \ln S^N_0 + \frac{1}{2} Q^N_n+O(1/\sqrt{N}).
   \end{align*}
    In terms of $\hat{Y}^N$, $\hat{M}^N$, and $\hat{Q}^N\set
    (Q^N_{[Nt]})_{0 \leq t \leq 1}$ this amounts to
    \begin{align*}
      \hat{Y}^N_t =
      \frac{\hat{M}^N_{t}}{\hat{S}^N_{t}}-\frac{\hat{M}^N_{0}}{\hat{S}^N_{0}}+\ln
        \hat{S}^N_t - \ln s_0+ \frac{1}{2}\hat{Q}^N_t+O(1/\sqrt{N})
    \end{align*}
    for $t \in \{0,1/N,\dots,1\}$. Since all other terms in this
    expression converge $\hat{\mathbb{P}}$-almost surely as $N \uparrow
    \infty$, so does $\hat{Q}^N$ and its limit is given by
    \begin{align*}
      \lim_N \hat{Q}^N_t = 2({Y}_t-\ln {M}_t+\ln s_0)
      =\langle Y\rangle_t
      =\langle \ln M\rangle_t \quad 0 \leq t \leq 1.
    \end{align*}
    Now, fix $0\leq s < t \leq 1$ and observe that
   \begin{align*}
     \langle \ln M \rangle_t - \langle \ln M \rangle_s &=\lim_N
     (\hat{Q}^N_{t}-\hat{Q}^N_{s}) \\
    & = \lim_N \sum_{[Ns]<n \leq [Nt]} \hat{X}^N_n\left(X^N_n+2\frac{M^N_n-S^N_n}{S^N_n}\right)\\
    & \in [\sigmalow(\sigmalow-2c)^+(t-s),\sigmahigh(\sigmahigh+2c)(t-s)].
   \end{align*}
  Hence, $\langle \ln M \rangle$ is absolutely continuous with density
  \begin{align*}
    \frac{d \langle \ln M \rangle_t}{dt} \in
  [\sigmalow(\sigmalow-2c)^+,\sigmahigh(\sigmahigh+2c)], \quad 0 \leq
  t \leq 1,
  \end{align*}
  which readily implies
  \begin{align*}
    a\left(\sqrt{\frac{d \langle \ln M \rangle_t}{dt}}\right) \leq c.
  \end{align*}
  It thus follows that $\mathbb{P} \set \law(M \,|\,
  \hat{\mathbb{P}})$ lies in the class $
  \mathcal{P}_{\sigmalow,\sigmahigh,c}$ as considered in
  Lemma~\ref{lem:101}~(i).  By the construction and
  $\hat{\mathbb{P}}$-almost sure convergence of $\hat{S}^M$,
  $\hat{M}^N$, $\hat{Q}^N$, this proves~\eqref{eq:21} for this
  $\mathbb{P}$.

  Let us now turn to the proof of item~(ii) of our lemma and take a
  $\overline{\sigma}:[0,1] \times C[0,1] \to \mathbb{R}$ from the class
  $\tilde{\Sigma}(c)$ introduced in Lemma~\ref{lem:101}. Fix $N \in
  \{1,2,\dots\}$ and define on $\Omega^N$ the processes $\sigma^N$,
  $\kappa^N$, $B^N$, and $\xi^N$ by the following recursion:
\begin{align*}
  \sigma^N_0 \set \sigmalow \vee \tilde{\sigma}_0(0) \wedge
  \sigmahigh, \;
   B^N_0 \set 0
\end{align*}
and, for $n=1,\dots,N$,
\begin{align*}
  \sigma^N_n &\set \sigmalow \vee \tilde{\sigma}_{(n-1)/N}(\overline{B}^N) \wedge
  \sigmahigh,\\
   \kappa^N_n &\set
   \frac{1}{2}\left(\frac{\tilde{\sigma}^2_{(n-1)/N}(\overline{B}^N)}{(\sigma^N_n)^2}-1\right),\\
   \xi^N_n &\set \sqrt{N} \frac{\ln S^N_n - \ln
     S^N_{n-1}}{\sigma^N_{n-1}}=\sqrt{N} X^N_n/\sigma^N_n,\\
   B^N_n &\set B^N_{n-1}+
\frac{\exp\left((1+\kappa^{N}_n)X^N_n-\kappa^{N}_{n-1}
X^N_{n-1}\right)-1}{\sqrt{1+2\kappa^{N}_n}\sigma^{N}_n}.
\end{align*}
Observe that the progressive measurability of $\tilde{\sigma}$ ensures
that its evaluation in the definition of $\sigma^N_n$ and $\kappa^N_n$
depends on $B^N=(B^N_m)_{m=0,\dots,N}$ only via its already constructed values for
$m=0,\dots,n-1$.

Next, define the process $q^N$ by
\begin{equation}\label{5.5}
q^N_n=\frac{\exp(\kappa^{N}_{n-1}X^{N}_{n-1})-\exp(-(1+\kappa^N_n)\sigma^{N}_n/\sqrt N)}
{\exp((1+\kappa^{N}_n)\sigma^{N}_n/\sqrt N)-\exp(-(1+\kappa^{N}_n)\sigma^{N}_n/\sqrt N)}.
\end{equation}
Consider the probability measure $\mathbb{P}^N$ on
$(\Omega^N,\mathcal{F}^N_N)$ for which the random variables
$\xi^{N}_1,...,\xi^{N}_N$ are i.i.d. with
$\mathbb{P}^N(\xi^{N}_1=1)=\mathbb{P}^N(\xi^{N}_1=-1)=1/2$. From
\eqref{3.3++} it follows that there exists $\epsilon>0$ for which
$\kappa^{N}_n>\epsilon-1/2$.  Thus $|B^{N}_n-B^{N}_{n-1}|=O(N^{-1/2})$
and also $|\kappa^{N}_n-\kappa^{N}_{n-1}|=O(N^{-1/2})$ because of the
Lipschitz continuity of $\tilde{\sigma}$.  We conclude that, for
sufficiently large $N$, $q^{N}_n\in (0,1)$ $\mathbb{P}^N$-almost
surely. For such $N$ we consider $\mathbb Q^N$ given by the
Radon-Nikodym derivatives
$$
\frac{d\mathbb{Q}^N}{d\mathbb{P}^N}|_{\mathcal{F}^{N}_n}=2^N\prod_{m=1}^n
\left(q^{N}_m\mathbb{I}_{\{\xi^{N}_m=1\}}+(1-q^{N}_m)\mathbb{I}_{\{\xi^{N}_m=-1\}}\right).
$$
Since $\mathbb{P}^N[\xi^N_n =\pm 1]=1/2$, our choice of
$q^N$~\eqref{5.5} ensures that $B^{N}$ is a martingale under $\mathbb
Q^N$.  Now consider the stochastic process
\begin{equation*}
M^{N}_n \set S^{N}_n\exp(\kappa^{N}_n X^N_n),
\quad n=0,\dots,N.
\end{equation*}
From~\eqref{3.3++}
it follows that $\kappa^{N}_N=0$ for sufficiently large $N$ and so
$M^{N}_N=S^{N}_N$. Furthermore, from \eqref{3.3+} we have
$\frac{|M^{N}_n-S^{N}_n|}{S^{N}_n}\leq \frac{c}{\sqrt N}$.
Observe also that
\begin{equation}\label{5.7}
M^{N}_n=M^{N}_{n-1}\left(1+\sqrt{1+2\kappa^{N}_n}\sigma^{N}_n(B^{N}_n-B^{N}_{n-1})\right).
\end{equation}
Hence, the predictability of $\sigma^{N}$, $\kappa^{N}$ ensures that,
along with $B^N$, also $M^N$
is a martingale under $\mathbb Q^N$. Hence $M^N$ and $\mathbb{Q}^N$
are as requested in part~(i) of our present lemma.

It thus remains to establish the weak convergence \eqref{eq:22}. By
applying Taylor's expansion we get
$$
\left|q^{N}_n-\frac{\kappa^{N}_{n-1}\sigma^{N}_{n-1}\xi^{N}_{n-1}+(1+\kappa^{N}_n)\sigma^{N}_n}
{2(1+\kappa^{N}_n)\sigma^{N}_n}\right|=O(N^{-1/2}) \quad
\mathbb{P}^N\text{-almost surely}
$$
Thus,
$\left|(2q^{N}_n-1)\xi^{N}_{n-1}-\frac{\kappa^{N}_{n-1}\sigma^{N}_{n-1}}{(1+\kappa^{N}_n)\sigma^{N}_n}\right|
=O(N^{-1/2})$ $\mathbb P^N$-almost surely.
From the last equality and the definition of the measure $\mathbb Q^N$ we get that
\begin{align*}
\mathbb E_{{\mathbb Q}^N}&\left[\left((1+\kappa^{N}_n)\sigma^{N}_n\xi^{N}_n
-\kappa^{N}_{n-1}\sigma^{N}_{n-1}\xi^{N}_{n-1}\right)^2\;|\;\mathcal{F}^{N}_{n-1}\right]
\\
=&(\sigma^{N}_n)^2(1+\kappa^{N}_n)^2+(\sigma^{N}_{n-1})^2(\kappa^{N}_{n-1})^2-2(2q^{N}_n-1)(1+\kappa^{N}_n)\sigma^{N}_n\kappa^{N}_{n-1}\sigma^{N}_{n-1}\xi^{N}_{n-1}\\
=&(\sigma^{N}_n)^2(1+2\kappa^{N}_n)+ O(N^{-1/2}).
\end{align*}
This together with applying the Taylor expansion yields
$$
\mathbb E_{\mathbb
  Q^N}\left[(B^{N}_n-B^{N}_{n-1})^2\;|\;\mathcal{F}^{N}_{n-1}\right]=\frac{1}{N}+O(N^{-3/2}).
$$
From Theorem~8.7 in~\cite{A} we get the convergence of
$\law(\overline{B}^{N} \;|\;\mathbb{Q}^N)$ to Wiener measure on $C[0,1]$.
From the continuity of $\tilde{\sigma}$ it follows that the continuous
interpolation on $[0,1]$ of
$(\{\sqrt{1+2\kappa^{N}_n}\sigma^{N}_n\})_{n=0}^N$, i.e., the process
$\tilde{\sigma}_{[Nt]/N}(\overline{B}^N)$ under $\mathbb{Q}^N$ converges
in law to $\tilde{\sigma}$ under $\mathbb{P}^W$ on $C[0,1]$.
Thus Theorem~5.4 in \cite{DP} and~\eqref{5.7} give the convergence
$$
\law(\overline{B}^{N}, \overline{M}^{N} \;|\; \mathbb{Q}^N)\Rightarrow \law(W,S^{\tilde{\sigma}}\;|\;\mathbb{P}^W)
$$
on the space $C[0,1]\times C[0,1]$.  Finally, observe that we have the
joint convergence
$$\law(\overline{S}^N,\overline{M}^N, \overline{\sqrt{N}|\kappa^{N}\sigma^{N}|}\, \  |\, \mathbb{Q}^N )
\Rightarrow \law(S^{\tilde{\sigma}},S^{\tilde{\sigma}},a(\tilde{\sigma})
\,|\,\mathbb{P}^W) \text{ as } N \uparrow \infty$$
on ${C}[0,1]\times C[0,1]\times C[0,1]$
and \eqref{eq:22} follows as required.
\end{proof}

\subsubsection{Proof of  `$\leq$' in~\eqref{eq:15}}\label{sec:leq}

Applying Theorem 2.2 with $g\set g^{N,c}$ of~\eqref{eq:10} shows that
for $N=1,2,\dots$ there exists a measure $\mathbb{Q}^N$ on
$(\Omega^N,\mathcal{F}^N_N)$ such that
\begin{equation}
  \label{eq:26}
  V^{\sigmalow/\sqrt{N},\sigmahigh/\sqrt{N}}_{g^{N,c}}(\F)
  \leq \frac{1}{N}+\mathbb{E}_{\mathbb{Q}^N}\left[\mathbb F\left(\overline{S}^N\right)-
\sum_{n=0}^{N-1}G^{N,c}_n\left(\frac{M^N_n-S^N_n}{S^N_n}\right)\right]
\end{equation}
where $G^{N,c}$ is the Legendre-Fenchel transform of $g^{N,c}$ and
where $M^N$ is defined as in Lemma~\ref{lem:102}. Since by
construction $h^c$ and, thus, also $g^{N,c}$ has maximum slope $c$, we
have
\begin{equation*}
  G^{N,c}_n(\alpha) =
 \begin{cases}
  H_{n/N}(\overline{S}^N,\alpha) & \text{ if } |\alpha| \leq c,\\
  \infty & \text{ otherwise.}
 \end{cases}
\end{equation*}
In particular, the above sequence of probabilities
$(\mathbb{Q}^N)_{N=1,2,\dots}$ is as required in the first part of
Lemma~\ref{lem:102}. We thus obtain the weak
convergence~\eqref{eq:21} with some probability $\mathbb{P}^* \in
\mathcal{P}_{\sigmalow,\sigmahigh,c}$.

Due to Skorohod's representation theorem there
exists a probability space
$(\hat{\Omega},\hat{\mathcal{F}},\hat{\mathbb{P}})$ with processes
$\hat{S}^N$, $\hat{M}^N$, and $\hat{Q}^N$, $N=1,2,\dots$ and a
continuous martingale $M>0$ such that
\begin{align}\label{eq:27}
        \law(\overline{S}^N,\overline{M}^N,\overline{Q}^N\, |\, \mathbb{Q}^N )
 &=       \law(\hat{S}^N,\hat{M}^N,\hat{Q}^N\, |\, \hat{\mathbb{P}} ), \quad
 N=1,2,\dots,\\\label{eq:28}
\law(S \,|\,\mathbb{P}^*) &= \law(M \,|\, \hat{\mathbb{P}}),
\end{align}
and such that
\begin{equation}\label{eq:29}
  (\hat{S}^N,\hat{M}^N,\hat{Q}^N) \to (M,M,\langle \ln M \rangle) \quad
  \hat{\mathbb{P}}\text{-almost surely as } N \uparrow \infty.
\end{equation}
Due to~\eqref{eq:19} of Lemma~\ref{lem:102}, $\max_{0 \leq t \leq 1}
\hat{S}^N_t$ is bounded in $L^p(\hat{\mathbb{P}})$ for any $p>0$. By
Lebesgue's theorem the assumed continuity and polynomial growth of
$\mathbb{F}$ in conjunction with~\eqref{eq:27} and~\eqref{eq:29}
thus suffices to conclude that
\begin{equation}
  \label{eq:30}
  \mathbb{E}_{\mathbb{Q}^N}\left[\mathbb F\left(\overline{S}^N\right)\right]
    \to \mathbb{E}_{\hat{\mathbb{P}}}\left[\mathbb  F\left(M\right)\right]
 \text{ as } N \uparrow \infty.
\end{equation}

Below we will argue that
\begin{align}
&\liminf_{N \uparrow \infty} \mathbb{E}_{\mathbb{Q}^N}\left[\sum_{n=0}^{N-1}G^{N,c}_n\left(\frac{M^N_n-S^N_n}{S^N_n}\right)\right] 
\geq \mathbb{E}_{\hat{\mathbb P}} \left[\int_{0}^1
\hat H_t(M)
a^2\left(\sqrt{\frac{d\langle\ln M\rangle_t}{dt}}\right) \,dt\right] \label{eq:32}.
\end{align}
Combining~\eqref{eq:26} with~\eqref{eq:30} and~\eqref{eq:32} then
gives
\begin{equation}
  \label{eq:33}
  \limsup_{N\uparrow \infty}
  V^{\sigmalow/\sqrt{N},\sigmahigh/\sqrt{N}}_{g^{N,c}}(\F) \leq \mathbb{E}_{\hat{\mathbb{P}}}\left[\mathbb  F\left(M\right)-\int_{0}^1
\hat H_t(M)
a^2\left(\sqrt{\frac{d\langle\ln M\rangle_t}{dt}}\right) \,dt\right].
\end{equation}
Because of~\eqref{eq:28} the right side of~\eqref{eq:33} can be
viewed as one of the expectations considered in~\eqref{eq:16}. We
deduce from Lemma~\ref{lem:101}~(i) that `$\leq$' holds
in~\eqref{eq:15}.

Let us conclude by proving~\eqref{eq:32} and write
\begin{align}\label{eq:34}
  \mathbb{E}_{\mathbb{Q}^N}\left[\sum_{n=0}^{N-1}G^{N,c}_n\left(\frac{M^N_n-S^N_n}{S^N_n}\right)\right]
  = \mathbb{E}_{\hat{\mathbb{P}}}\left[\int_0^1  NH_{[Nt]/N}\left(\hat{S}^N,\Delta^N_t/\sqrt{N}\right)\,dt\right]
\end{align}
where
\begin{align*}
 -c \leq \Delta^N_t \set \sqrt{N}
 \frac{\hat{M}^N_{[Nt]/N}-\hat{S}^N_{[Nt]/N}}{\hat{S}^N_{[Nt]/N}} \leq c.
\end{align*}
Let $b:\mathbb R\rightarrow\mathbb R$ be a convex function such that
$b(u)=a^2(\sqrt u)$ for $u\geq 0$, for instance
\begin{equation}\label{eq:35}
  b(u) \set
  \begin{cases}
    -u, & u\leq-\sigmalow^2,\\
    \frac{1}{4}\frac{(\sigmalow^2-u)^2}{\sigmalow^2}, & -\sigmalow^2<u<0,\\
    a(\sqrt u)^2, & u\geq 0.
  \end{cases}
\end{equation}
By Assumption~\ref{quadratic} and the $\hat{\mathbb{P}}$-almost sure
convergence of $\hat{S}^N$ to $M$, the integrand
$NH_{[Nt]/N}\left(\hat{S}^N, \Delta/\sqrt{N}\right)$ converges
uniformly in $t \in [0,1]$ and $\Delta \in [-c,c]$ to $\hat{H}_t(M)
\Delta^2$. Moreover, from Lemma \ref{lem3.17} below we have
$(\Delta^N_t)^2 \geq b({N\Delta\hat{Q}^N_{[Nt]/N}})$. Hence, the
$\liminf$ in~\eqref{eq:32} is bounded from below by
\begin{align}\label{eq:36}
  \liminf_{N \uparrow \infty}   \mathbb{E}_{\hat{\mathbb{P}}}\left[\int_0^1
    \hat{H}_t(M) (\Delta^N_t)^2\,dt\right] \geq
  \liminf_{N \uparrow \infty}   \mathbb{E}_{\hat{\mathbb{P}}}\left[\int_0^1
    \hat{H}_t(M) b(N{\Delta\hat{Q}^N_{[Nt]/N}})\,dt\right].
\end{align}
From their definition it follows that the processes $(N\Delta
\hat{Q}^N_{[Nt]/N})_{0 \leq t \leq 1}$, $N=1,2,\dots$ take values in
the interval $[\sigmalow^2-2c\sigmahigh,\sigmahigh^2+2c\sigmahigh]$,
and so they are bounded uniformly. Hence, we can use Lemma~A1.1
in~\cite{DS2} to find, for $N=1,2,\dots$, convex combinations
$\delta^N$ of elements in this sequence of processes with index in
$\{N, N+1, \dots\}$ which converge $\hat{\mathbb{P}}\otimes dt$-almost
every. Denote the limit process by $(\delta_t)_{0 \leq t \leq 1}$.  Observe
that
$$\int_{0}^{t} \delta_u du=\lim_{N\rightarrow\infty}\int_{0}^{t} N\Delta \hat{Q}^N_{[Nu]/N}du=
\langle \ln M \rangle_{t}$$
and so we conclude that
$\delta_t = d\langle \ln M \rangle_t/dt$,
 $\hat{\mathbb{P}}\otimes dt$-almost every.
Hence, by
Fatou's lemma,
\begin{align*}
 \mathbb{E}_{\hat{\mathbb{P}}}\left[\int_0^1 \hat{H}_t\left(M\right)
   a^2(\sqrt{\delta_t}) \,dt\right]
  &=\mathbb{E}_{\hat{\mathbb{P}}}\left[\int_0^1 \hat{H}_t\left(M\right)
     b({\delta_t}) \,dt\right]\\
       &\leq
  \liminf_{N \uparrow \infty} \mathbb{E}_{\hat{\mathbb{P}}}\left[\int_0^1
    \hat{H}_t\left(M\right)
      b({\delta^N_t}) \,dt\right].
\end{align*}
By convexity of $b$ and by the construction of
$\delta^N$ this last $\liminf$ is not larger than the $\liminf$ on the
right side of~\eqref{eq:36}. Hence, we can combine these estimates
to obtain our assertion~\eqref{eq:32}.

\subsubsection{Proof of  `$\geq$' in~\eqref{eq:15}}\label{sec:geq}

Because of Lemma~\ref{lem:101} it suffices to show that
\begin{equation}
  \label{eq:37}
 \liminf_{N \uparrow \infty} V^{\sigmalow/\sqrt{N},\sigmahigh/\sqrt{N}}_{g^{N,c}}(\F)
  \geq \mathbb{E}_{\mathbb{P}^W}
  \left[\mathbb F(S^{\tilde{\sigma}})-\int_{0}^1\hat H_t(S^{\tilde{\sigma}})a^2(\tilde{\sigma}_t) dt\right]
\end{equation}
for $\tilde{\sigma} \in \tilde{\Sigma}(c)$. For such $\tilde{\sigma}$, we can apply
Lemma~\ref{lem:102}(ii) which provides us with probabilities
$\mathbb{Q}^N$ and martingales $M^N$ on $(\Omega^N,\mathcal{F}^N_N)$, $N=1,2,\dots$, such
that
\begin{equation}
  \label{eq:38}
        \law\left(\left.\overline{S}^N,\overline{M}^N,\overline{\sqrt{N}\frac{|M^N-S^N|}{S^N}}\,
            \right|\, \mathbb{Q}^N\right)
      \Rightarrow \law\left(\left.S^{\tilde{\sigma}},S^{\tilde{\sigma}},a(\tilde{\sigma})
      \,\right|\,\mathbb{P}^W\right) \text{ as } N \uparrow \infty.
\end{equation}
We can now use Skorohod's representation theorem exactly as in
Section~\ref{sec:leq} to obtain, in analogy with~\eqref{eq:30}, that
\begin{equation}
  \label{eq:39}
  \mathbb{E}_{\mathbb{Q}^N}\left[\mathbb F\left(\overline{S}^N\right)\right]
    \to \mathbb{E}_{\mathbb{P}^W}\left[\mathbb  F\left(S^{\tilde{\sigma}}\right)\right]
 \text{ as } N \uparrow \infty.
\end{equation}
Set
\begin{align*}
 -c \leq \Delta^N_t \set \sqrt{N}
 \frac{\hat{M}^N_{[Nt]/N}-\hat{S}^N_{[Nt]/N}}{\hat{S}^N_{[Nt]/N}} \leq c.
\end{align*}
In analogy with~\eqref{eq:34} we have that
\begin{align}\label{eq:40}
 \lim_{N \uparrow \infty}&
 \mathbb{E}_{\mathbb{Q}^N}\left[\sum_{n=0}^{N-1}G^{N,c}_n\left(\frac{M^N_n-S^N_n}{S^N_n}\right)\right]\\\nonumber
 &=\lim_{N \uparrow \infty}\mathbb{E}_{\mathbb{Q}^N}\left[\int_0^1
    NH_{[Nt]/N}\left(\hat{S}^N,\Delta^N_t/\sqrt{N}\right)\,dt\right]\\\nonumber
  &= \mathbb{E}_{\mathbb{P}^W}[\int_{0}^1\hat H_t(S^{\tilde{\sigma}})a^2(\tilde{\sigma}_t) dt],
\end{align}
where the last two equalities follow from our assumptions on $H$, the
moment estimate in Lemma~\ref{lem:102} (which gives uniform integrability)
and~\eqref{eq:38} (together with the Skorohod representation
theorem).

Combining~\eqref{eq:39} with~\eqref{eq:40} then allows us to write
the right side of~\eqref{eq:37} as
\begin{align*}
  \mathbb{E}_{\mathbb{P}^W} & \left[\mathbb F(S^{\tilde{\sigma}})-\int_{0}^1\hat  H_t(S^{\tilde{\sigma}})a^2(\tilde{\sigma}_t) dt\right] \\
  &=
  \lim_{N\uparrow \infty} \mathbb{E}_{\mathbb{Q}^N}\left[\mathbb F\left(\overline{S}^N\right)-\sum_{n=0}^{N-1}G^{N,c}_n\left(\frac{M^N_n-S^N_n}{S^N_n}\right)\right].
\end{align*}
On the other hand, Theorem~\ref{thm2.1} reveals that such a limit is a
lower bound for the left side of~\eqref{eq:37}. This accomplishes our
proof that `$\geq$' holds in~\eqref{eq:15} \qed

We complete the section with the following elementary inequality which
we used in the previous proof:
\begin{lem}\label{lem3.17}
For $x,y\in\mathbb R$ such that
$|x|\in [\sigmalow,\sigmahigh]$ we have
$$b(x^2+2xy)\leq y^2,$$
where $b$ denotes the function of~\eqref{eq:35}.
\end{lem}
\begin{proof}
We distinguish four cases:
\begin{itemize}
\item[(i)] $x^2+2xy\geq\sigmahigh^2$.  In this case $0\leq
  x^2+2xy-\sigmahigh^2\leq 2xy\leq 2\sigmahigh |y|$ and so
$$b(x^2+2xy)=\frac{1}{4}\frac{(x^2+2xy-\sigmahigh^2)^2}{\sigmahigh^2}\leq
y^2.$$

\item[(ii)] $\sigmalow^2\leq x^2+2xy\leq \sigmahigh^2$. In this case $b(x^2+2xy)=0$ and the statement is trivial.

\item[(iii)] $|x^2+2xy|\leq\sigmalow^2$.  Set $z=x^2+2xy$. Assume that
  $z$ is fixed and introduce the function $f_z(u)=\frac{(z-u)^2}{u}$,
  $u\geq\sigmalow^2$.  Observe that for $z\in
  [-\sigmalow^2,\sigmalow^2]$ the derivative $f_z'(u)=1-z^2/u^2\geq 0$
  and so
$$b(x^2+2xy)=\frac{1}{4}f_z(\sigmalow^2)\leq \frac{1}{4} f_z(x^2)=y^2.$$
\item[(iv)] Finally, assume that $x^2+2xy\leq -\sigmalow^2$. Clearly
  $x^2+2xy +y^2\geq 0$ and so
$$b(x^2+2xy)=-x^2-2xy\leq y^2.$$
\end{itemize}
\end{proof}


\end{document}